\documentclass[10pt, conference, compsocconf]{IEEEtran}
\usepackage{graphicx}
\usepackage{hyperref}
\usepackage{algorithmic}
\usepackage[ruled,vlined,linesnumbered]{algorithm2e}
\usepackage{epsfig}
\usepackage{subfig}
\usepackage{balance}  
\usepackage[table]{xcolor}
\usepackage{color,soul,xspace}
\usepackage{amssymb, amsmath, amsthm, amsfonts}

\newcommand{\eat}[1]{}

\newtheorem{thm}{Theorem}
\newtheorem{cor}[thm]{Corollary}
\newtheorem{lemma}{Lemma}
\newtheorem{definition}{Definition}

\newcommand{\vsa}{\vspace*{-0.2cm}}
\newcommand{\vsb}{\vspace*{-0.2cm}}
\newcommand{\vsc}{\vspace*{-0.4cm}}

\clearpage
\begin{document}
%



\title{Towards Scalable Network Delay Minimization}

\author{\IEEEauthorblockN{Sourav Medya}
\IEEEauthorblockA{University of California, Santa Barbara\\
medya@cs.ucsb.edu}
\and
\IEEEauthorblockN{Petko Bogdanov}
\IEEEauthorblockA{University at Albany - SUNY\\
pbogdanov@albany.edu}
\and
\IEEEauthorblockN{Ambuj Singh}
\IEEEauthorblockA{University of California, Santa Barbara\\
ambuj@cs.ucsb.edu}
}



%
%
%
%

\maketitle
\begin{abstract}

Reduction of end-to-end network delays is an optimization task with applications in multiple domains. Low delays enable improved information flow in social networks, quick spread of ideas in collaboration networks, low travel times for vehicles on road networks and increased rate of packets in the case of communication networks. Delay reduction can be achieved by both improving the propagation capabilities of individual nodes and adding additional edges in the network. One of the main challenges in such design problems is that the effects of local changes are not independent, and as a consequence, there is a combinatorial search-space of possible improvements. Thus, minimizing the cumulative propagation delay requires novel scalable and data-driven approaches.

In this paper, we consider the problem of network delay minimization via node upgrades. Although the problem is NP-hard, we show that probabilistic approximation for a restricted version can be obtained. We design scalable and high-quality techniques for the general setting based on sampling and targeted to different models of delay distribution. 
Our methods scale almost linearly with the graph size and consistently outperform competitors in quality.\\


\end{abstract}
\section{Introduction}

Given a communication network, how can one minimize the end-to-end communication delay by upgrading networking devices? How to minimize the travel time on an airline network by increasing the personnel and infrastructure at key airports? How to recruit users who can quickly re-post updates enabling fast global propagation of information of interest in a social network? There is a common \emph{network design problem} underlying all the above application scenarios: for a large network with associated node delays, identify a set of nodes (within budget) whose delay reduction will minimize the path delays between any pair of nodes. 

Network design problems, including planning, implementing and augmenting networks for desirable properties, have a wide range of applications in communication, transportation and information networks as well as VLSI design~\cite{gupta2011approximation,o1994hub,lin2015,zhu2004power, Khalil2014, saha2015approximation}. 
Challenges in this area are posed by the rapidly growing sizes of real-world networks, leading to the need for scalable, data-driven approaches.  
In particular, network design problems involve local changes to an existing large network such as adding/modifying links or nodes as a means to improve its global properties~\cite{chaoji2012recommendations, Tong2012GML, dilkina2011, meyerson2009, lin2015,Khalil2014}. In this paper we address a problem from the above category, namely, minimizing the overall end-to-end network delay. 

The end-to-end delay in a network affects propagation speeds and is a function of the network link connectivity and the throughput capabilities of individual nodes. The majority of previous work focuses on delay minimization by augmenting network edges~\cite{paik1995, krumke1998,meyerson2009,lin2015,saha2015approximation}. Less attention has been devoted to the complementary, but algorithmically non-equivalent setting in which the propagation capabilities of individual nodes are ``upgraded'' under budget~\cite{dilkina2011}. In this paper we address the node-version of the delay minimization problem. A toy example instance and possible solutions for the problem are presented in Fig.~\ref{fig:example}. All nodes start with a delay of $1$. The objective is to select a set of two nodes whose delay reduction will minimize the overall network delay (sum of pairwise delays). For instance, selecting nodes $c,d$ (Fig.~\ref{fig:T2}) is a better solution than selecting nodes $a,f$ (Fig.~\ref{fig:T1}).


\begin{figure}[t]
    \centering
    \includegraphics[keepaspectratio, width=.43\textwidth]{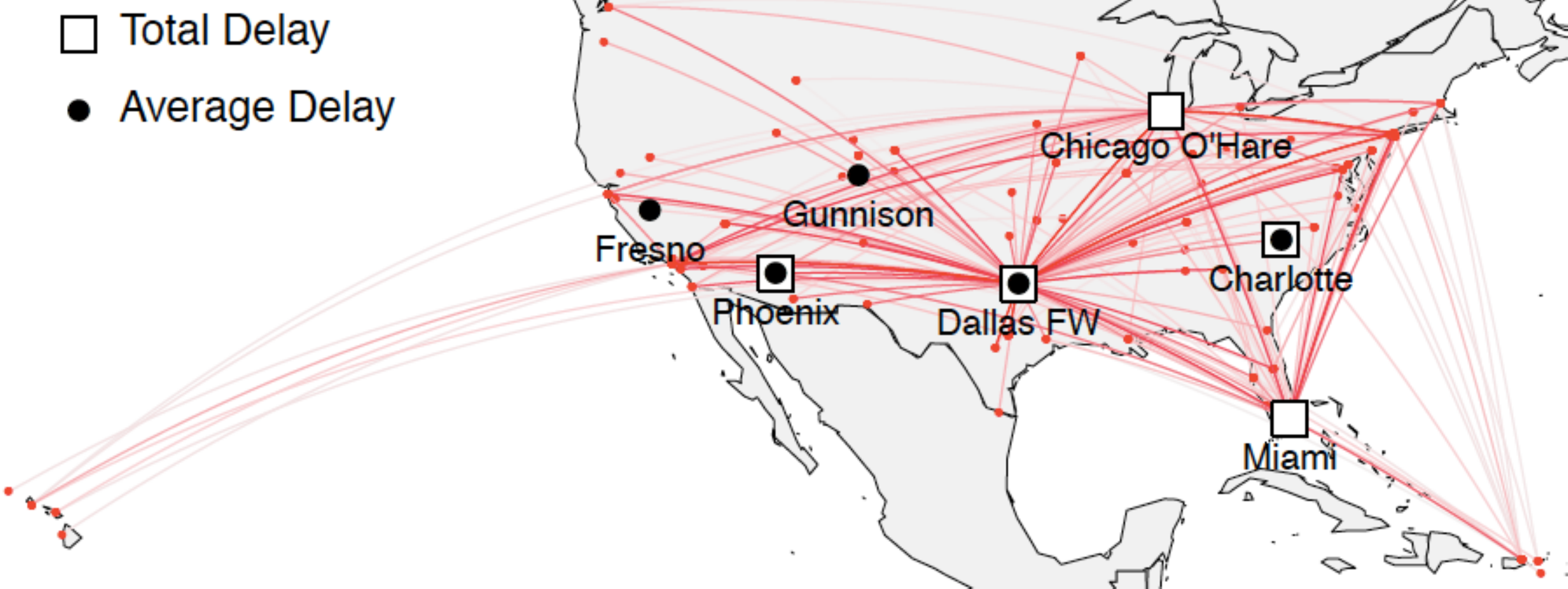}
    \caption{ \small Airports with maximum impact on the overall network delay in the \textit{American Airlines} network as discovered by our methods. If airline-caused delays are removed in these airports, the overall network delay decreases by 96\% and 55\% when the accumulated airline-caused flight delays  (\textit{Total}) and \textit{Average} (over number of flights) airport flight delays are considered respectively (data from US Dept. of Transportation).}\label{fig:AA}
     \vspace{-6mm}
\end{figure}
Depending on the domain, low end-to-end delay enables improved information flow in social and collaboration networks~\cite{Cha2009}, reduced travel time for airline and road networks~\cite{AhmadBeygi2008} and increased throughput for communication networks~\cite{dilkina2011}. Consider, for example, the air transportation network of a major US carrier presented in Fig.~\ref{fig:AA}, where edges correspond to flights offered by the carrier between endpoint cities. Based on historical information on past flights one can associate airports with airline-caused delays. An important question for an airline is then how to minimize overall delays by improving the number of personnel and available infrastructure (e.g. luggage handling) in problematic airports that affect multiple routes. In Fig.~\ref{fig:AA} we show the airports with highest delay-reduction potential, determined based on both historical delays and their position in the network. When the cumulative historical delays are considered (\textit{Total}), hub airports like Chicago, Dallas and Miami constitute the best solution, while ``fringe'' airports make it to the list when the \textit{Average} delay is considered\footnote{Extended discussion of the findings in this data is available in the Experimental section and the Appendix.}.

Another important application comes from social networks where user behavior---activity and interest in a specific topic---determines the node delay for information propagation. In this domain, the objective is to speed up the global propagation of information by decreasing individual response time~\cite{liu2012time}. A social media strategist of an election campaign, for example, would be interested in recruiting social network users who can re-post relevant campaign updates immediately, enabling faster propagation of relevant campaign information. Both the position in the network and the current delay in propagating information should be taken into account in selecting recruits. While information and influence propagation are traditionally modelled as a diffusion process (i.e., using all possible paths)~\cite{kempe2003maximizing}, multiple recent approaches (including the current work) focus on the most probable (shortest) paths in order to allow scalable solutions~\cite{kimura2006tractable,chen2010scalable}.

Given a network with node delays, our goal is to identify a set of nodes whose delay reduction will minimize the sum of shortest path delays between pairs of nodes. We term this problem the \emph{Delay Minimization Problem (DMP)}, and demonstrate that it is NP-hard in a general network, even under the assumption that node delays are uniform. Intuitively, the challenge stems from the fact that the global effect of a single node upgrade is dependent on the remaining nodes in the solution. We provide approximation analysis for special cases and characterize the effectiveness of randomized approaches based on Vapnik-Chervonenkis theory~\cite{vapnik1971}. 

We propose a \emph{Greedy} approach which is optimal for restricted graph structures and has a good quality in practice for general graphs. However, Greedy does not scale to large networks due to its high computational complexity. We develop sampling-based alternatives that make similar selection of node upgrades with high probability using knowledge of only a small fraction of the network. The underlying delay model, uniform or arbitrary delays, plays an important role in the design and solution quality of our corresponding sampling schemes. We provide theoretical and experimental analysis of our algorithms for different network structures and node delay distributions and demonstrate their real-world utility.


Our contributions in this paper include: \\
\noindent $\bullet$ We consider the node delay minimization problem and show that it is NP-hard even for uniform delays. We show an approximation guarantee by 
random schemes using VC dimension theory and analyze the complexity for a restricted problem formulation.\\
\noindent $\bullet$ We propose high-quality sampling-based algorithms that scale almost linearly with the network size. In million-node networks we obtain high-quality solutions within $1$ hour, while non-trivial alternatives are infeasible.\\
\noindent $\bullet$ We show that the solutions produced by our algorithms in real-world datasets are consistently better than those of competitors (up to $10\%$ higher delay reduction). \\
%





\vsc
\section{Problem Definition}
\label{sec:problem_definition}

\begin{figure}[t]
\vspace{-2mm}
    \centering
    \vsa
    \subfloat[\{a,f\}, SPD=48]{
        \includegraphics[ keepaspectratio, width=.16\textwidth]{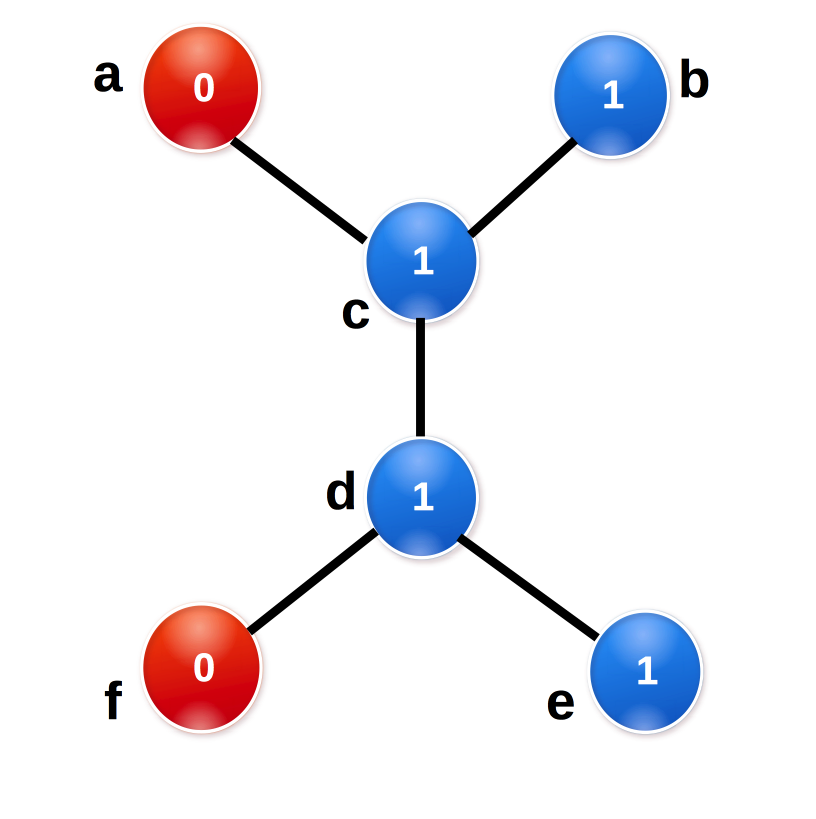}
       \label{fig:T1}
    }
    \subfloat[\{c,d\}, SPD=20]{
       \includegraphics[ keepaspectratio, width=.16\textwidth]{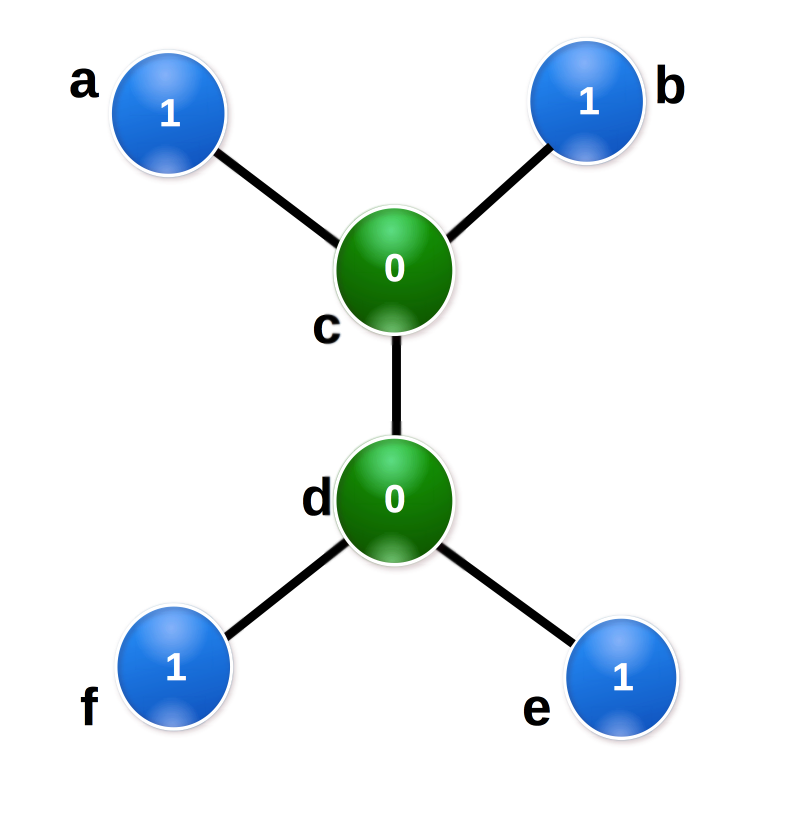}
       \label{fig:T2}
    }
       
     \caption{\small Illustrative example of reduction of all pair shortest path delays by reducing the delays of two nodes (budget of 2). Initially every vertex has a delay of $1$. (a) and (b) represent examples of a non-optimal and the optimal Target Set selection respectively.}\label{fig:example}
     \vspace{-4mm}
\end{figure}

A network is modeled as an undirected graph $G(V,E,L)$, where $V$ and $E$ are sets of vertices and edges respectively and $L$ is a function $L:V\rightarrow\mathbb{R}_{> 0}$ over $V$ that specifies the delay/latency $l(v)$ of individual nodes. 
The delay (or length) of a path is defined as the cumulative delay of the vertices along the path, excluding that of the destination. 
More formally, if $P_{s,t}=(v_{s},v_{1},v_{2},...,v_{r},v_{t})$ is a path from vertex $v_{s}$ to $v_{t}$, its length is defined as $l(P_{s,t}) = l(v_{s})+\Sigma_{i=1}^{r} l(v_{i})$. 
Delay at the destination node in a path is excluded since our targeted applications consider information/traffic flow and the destination node does not add any delays. 
The shortest path between vertices $s$ and $t$ is that of minimum length (delay) among all such paths and its length is denoted as $d(s,t)$. By convention, $d(s,s)=0$ for all $s\in V$. We define the \textit{all pair shortest path delays (SPD)} as the sum of shortest path lengths between all pairs of vertices, i.e., $SPD(G)=\Sigma_{s,t\in V}d(s,t)$. 


The DMP asks for a subset of vertices whose upgrade (delay reduction) minimizes the overall \textit{SPD}. In the process, the delay of a fixed (small) number of vertices $T\subset V$ is reduced to $0$\footnote{
Reduction by units of delay can be approached with simple changes in our algorithms.}. We call this subset $T$ a \textit{Target Set (TS)} and its size $|T|=k$, the budget. The upgrade of the TS, $T$ results in reduction of the lengths of shortest paths in the network. We denote the resulting (effective) shortest path length between $s$ and $t$ given the upgrade of $T$ as $d(s,t|T)$. Our goal is to find a $T$ that minimizes $\Sigma_{s,t\in V}d(s,t|T)$.

\vspace{-1mm}
\begin{definition}\textbf{Delay Minimization Problem (DMP):}
\vspace{-1mm}
Given a network $G=(V,E,L)$ and a budget $k$, find a target set $T\subset V$, such that $|T|=k$ and $\Sigma_{s,t\in V}d(s,t|T)$ is minimized.
\label{def:DMP}
\vspace{0mm}
\end{definition}
Fig.~\ref{fig:example} shows two possible TS solutions of size $k=2$ for a small example network. Initially all vertices have a delay of $1$ corresponding to an \textit{SPD} of $\Sigma_{s,t\in V}d(s,t)=58$. The reduction due to any TS $T$ is defined as the difference between the initial and the upgraded \textit{SPD}, i.e. $\Sigma_{s,t\in V}d(s,t)-\Sigma_{s,t\in V}d(s,t|T)$. An optimal TS maximizes the reduction (and minimizes the upgraded SPD). Thus, Fig.~\ref{fig:T1} shows a sub-optimal TS $\{a,f\}$ with reduction of $10$, while Fig.~\ref{fig:T2} shows an optimal TS $\{c,d\}$ with maximum \textit{SPD} reduction of $38$. Our goal is to minimize the \textit{SPD} by finding the optimal TS of budget size at most $k$.

\subsection{Complexity}
\label{sec:hardness}
We consider two different models for the distribution of the delays in a network and characterize the problem complexity. Under the \textit{general model}, node delays can be arbitrary non-negative values, while the \textit{uniform model} assumes equal delays (for simplicity, delay of $1$) on all nodes.
We show that DMP is NP-hard in the special case of the \textit{uniform model}, and hence it is in the same complexity class as the \textit{general model}. To show this hardness result we reduce the Set Cover problem to our problem. However, for restricted network structures such as trees, finding an optimal TS takes polynomial time.

\begin{thm}\label{thm:nphard} DMP is NP-hard even if the delay of all vertices is $1$, i.e. under the uniform model (or general model). 
\end{thm}

\begin{proof}
See the Appendix.
\end{proof}


Theorem \ref{thm:nphard} establishes that the problem is NP-hard and finding an optimal TS of size $k$ would involve enumerating all $O(|V|^{k})$ subsets of $V$. However, finding an optimal TS in trees takes polynomial time under the general model. Shortest paths between any pair of nodes in trees are unique and, hence, they do not change after upgrading the delay of any vertex. Intuitively, this fact about trees helps a simple greedy algorithm (formally defined as Algorithm \ref{algo:GR}) to produce an optimal TS of size $k$. 
\subsection{Approximability}
\label{sec:approx}
Since DMP is NP-hard, we explore the existence of approximations with guarantees. Maximizing a non-negative, monotone and submodular~\cite{nemhauser1978} function using a greedy approach leads to a well known constant time approximation of $(1-1/e)$. The underlying objective function in DMP is monotone as the $SPD$ reduces after each upgrade. However, it does not have the submodular property.
 
\begin{lemma}
\vspace{-2mm}
The objective function in DMP is monotone but not submodular, even under the uniform model.
\vspace{-1mm}
 \end{lemma}
\begin{proof} The objective function $f(T)$ in DMP  is ``delay reduction'' defined as $f(T)=\Sigma_{s,t\in V}d(s,t)-\Sigma_{s,t\in V}d(s,t|T)$, where $T$ is the target set. The function $f(T)$ is monotone in the size of TS. To prove non-submodularity, we consider the example of a ring graph $G$ of six vertices with unit delays: vertex $x_1$ is connected to $x_2$, $x_2$ to $x_3$ and so on, and finally $x_6$ is connected to $x_1$.  The intuition is the following: a super-set of nodes as TS might force more shortest paths through the newly added vertex than its sub-set as TS. Let set $A = \phi$, $B=\{x_2,x_4\}$. 
In our example, $f(B\cup\left\{ x_3\right\} )=54-21=33$, $f(B)=54-34=20$, $f(A\cup\left\{ x_3\right\} )=54-43=11$, $f(A)=0$. So, $f(B\cup\left\{ x_3\right\} )-f(B) > f(A\cup\left\{ x_3\right\} )-f(A)$. So, $f(\cdot)$ is not submodular.
\vspace{-1mm}
\end{proof}

We show that there exists an approximation for a restricted variant of DMP targeting long paths for delay reduction. Focusing on long paths, as opposed to reduction across all lengths, is useful in applications where delays up to a given threshold do not affect the overall system operation. For example, participants in a multi-way video conference need to receive frames in at most $0.1s$ to ensure good video quality, but improving the delay for pairs that meet this requirement does not provide further benefit. To model this, instead of SPD, we consider the sum of \textit{``long''} shortest paths as an optimization metric. We achieve a probabilistic approximation of $O(k)$ based on VC-dimension theory~\cite{vapnik1971} and random sampling. We exploit a relationship between the set of shortest paths in a network and their VC-dimension to prove the stated approximation (see the appendix).



\section{Algorithms}
\label{sec:greedy}
We present a greedy approach for DMP that consecutively selects the vertex that minimizes the SPD in each iteration. Such an approach is optimal for $k=1$. It also produces optimal results for networks with simple structures (Lemma \ref{lemma:GR_structure}) and works well in practice for general instances. It is, however, expensive as it requires re-computation of all shortest paths at every iteration. To make the approach scalable, we employ sampling techniques and introduce probabilistic approximation algorithms for different delay models.

\subsection{Greedy Construction of the Target Set}
\label{sec:greedyintro}
While finding the optimal TS is NP-hard, in the case of only one target vertex, an exact solution can be obtained by computing the reduction of all individual nodes in polynomial time. Therefore, a greedy algorithm, which selects a vertex that optimally reduces SPD at each step, is a natural approach to solve DMP. Before presenting the algorithm, we introduce some additional notation. We define the delay \textit{Reduction (RS)} by a target set $S$ as: 
\begin{equation*}
\small
RS(S)=\Sigma_{s,t\in V}d(s,t)-\Sigma_{s,t\in V}d(s,t|S).
\end{equation*}

We further define RS by a vertex $v$, given that a subset $S$ has already been included in TS (assuming $v\notin S$) as:
\begin{equation*}
\small
RS(v|S)=\Sigma_{s,t\in V}d(s,t|S)-\Sigma_{s,t\in V}d(s,t|S\cup\{v\}).
\vspace{-1mm}
\end{equation*}

The reduction of adding vertex $v$ to a set $S$ in TS can be expressed as $RS(S\cup\{v\})=RS(v|S)+RS(S)$.
The $RS$ of a vertex depends on: (i) its delay and (ii) the number of unique shortest paths passing through it after removing its delay. Maximizing $RS(v|S)$ takes both these properties into account. Next, we present an algorithm which iteratively selects the vertex of maximum reduction $RS(v|S)$.


GR (Alg. \ref{algo:GR}) is a greedy TS selection strategy. It takes a network $G$ ($|V|=n$ and $|E|=m$) and a budget $k$ as input. First it pre-computes all pairs of shortest paths and stores them in an $n\times n$ matrix $A$ (steps 2-3). Then it computes the TS of $k$ vertices in $k$ iterations. In each iteration, it selects the vertex with maximum $RS$ conditioned on the current TS (step 5-8). When probing each vertex, the algorithm assumes its delay as $0$, updates the stored distances accordingly and estimates the reduction of the vertex. It chooses the vertex of maximum $RS$, makes its delay permanently $0$, and adds it to the TS. GR also updates the stored shortest path distances accordingly. 
GR is optimal for certain families of networks with simple structure. The following lemma outlines such families.

\begin{algorithm}[t]
\scriptsize
 \caption {Greedy (GR)\label{algo:GR}}
\begin{algorithmic}[1] 
 \REQUIRE Network $G=(V,E,L)$, Vertex delays $l(v)$, Budget $k$
 \ENSURE A subset of $k$ nodes 
\STATE Initialize Matrix $A$ with $0$ and $T$ as $\Phi$ 
\STATE Compute all pair shortest paths
\STATE Store $d(s,t)$ in Matrix position $A_{s,t}$ 
\WHILE {  $|T|\leq k$ }
\FOR{$v' \in V$}
\STATE Compute $RS(v'|T)$ when $l(v')>0$
\ENDFOR
\STATE $v\leftarrow max_{v'\in V}\{RS(v'|T)\}$ and then set $l(v)$ as $0$
\STATE Update $d(s,t)$ for $s,t \in V$ as $l(v)$ becomes $0$ 
\STATE $T\leftarrow T\cup \{v\}$
\ENDWHILE
\STATE Return $T$
\end{algorithmic}
\end{algorithm}

\begin{lemma}\label{lemma:GR_structure}
\vspace{-1mm}
Greedy (Alg. \ref{algo:GR}) produces an optimal TS in restricted structures such as trees, cliques and complete bipartite graphs under the general model.
\vspace{-1mm}
\end{lemma}

In the previous section (Section \ref{sec:hardness}), we already discussed why GR produces an optimal result for trees. In a clique, since there is an edge between any pair of vertices, selecting $k$ vertices in descending delay order produces an optimal result and this is exactly the selection of GR. In a complete bipartite graph, if the delay of one vertex is updated then all vertices from the opposite partition will use this vertex to reach other vertices in their own partition, hence the greedy selection will again produce an optimal solution.  

\noindent \textbf{Example:} We provide a running example of GR in Fig.~\ref{fig:T2}. The first selected vertex is either $c$ or $d$ as $RS(c)=RS(d)=19$ and the $RS$ of any other vertex is $5$. Assuming that GR chooses $c$ at the first step, the next best vertex is $d$ as $RS(d|\{c\})=19$. $RS(v|\{c\})$, when $v$ is any other vertex, remains $5$. In the example, GR produces the optimal TS as the network structure is a tree. 



\noindent \textbf{Complexity:} GR runs in time $O(kn^3)$ which is dominated by the computation of shortest paths in steps $2,6$ and $9$. Finding the next ``best" vertex by evaluating the reduction of all possible vertices requires $O(n^3)$ time, where $n$ is the number of vertices. Moreover, updating the distances after a vertex is included in TS takes $O(n^2)$. The space complexity of computing all pairs shortest paths is $O(n^2)$. The high complexity of GR introduces a scalability challenge, rendering the algorithm infeasible for large real-world networks. Hence, we develop sampling-based versions of GR for large graphs and provide approximation guarantees w.r.t. GR.

\subsection{General Model: Approximate Target Set}


The main drawback of GR is that it is not scalable. We address its computational and storage bottlenecks using a sampling scheme. The main idea behind our approach is as follows: instead of computing and optimizing the sum of distances between all pairs of vertices, we can estimate it based on a small number of sampled vertex pairs. 

In what follows, we bound the difference in quality of our sampling solution GS (presented in Alg.~\ref{algo:GS}) and Greedy (Alg.~\ref{algo:GR}). In this case, the absolute value of the reduction $RS$ is not a suitable metric as the initial sum of shortest path distances (SPD) varies across input graphs. Hence, we choose \textit{Relative Reduction (RR)} as a quality metric where we normalize $RS$ by the initial SPD. We define the measure \textit{RR} of a set $S$ as $RR(S)  =\frac{RS(S)}{SPD}$. The \textit{RR} of a vertex $v$ given a set $S$ comprising the current TS is defined in a similar manner, $RR(v|S)=\frac{RS(v|S)}{SPD}$. 

As part of GS, we sample uniformly with replacement a set of ordered vertex pairs $P$ of size $p$ ($|P|=p$) from the set of all vertex pairs $U=\{(s,t)|s \in V,t\in V, s\neq t\},|U|=n(n-1)$. 
The samples can be viewed as random variables associated with the selection of a pair of vertices and the distance between a sampled pair is the value of the random variable. When uniform random sampling is used, each pair is chosen with probability $\frac{1}{n(n-1)}$ and the choice of one sample does not affect that of any other sample. Thus, the samples are independent and identically distributed random variables.

We first show that the estimate of SPD based on samples is unbiased. Namely, for any target set of nodes $S$, the average of the sum of distances between pairs in $P$ is an unbiased estimate of that between all pairs of vertices, the latter being defined as $\mu=\frac{\Sigma_{s,t\in V}d(s,t|S)}{n(n-1)}$. The vertex whose inclusion in TS optimizes this estimate is chosen in each step of GS. 

\begin{lemma}\label{lemma:means_general}
Given a sample of node pairs $P, |P|=p$, the expected average distance among the sampled pairs is an unbiased estimate of the average of all-pair distances ($\mu$): 
$E[\frac{1}{p}\sum\limits_{i=1}^pX_i]= \mu$
where $X_i$ represents the distance between the $i$-th pair of vertices in the sample. 
\end{lemma}
\begin{proof}
The random variable, $X_i$ is the the distance between the $i$-th pair of vertices in the sample. The probability of a pair in the selection is $\frac{1}{n(n-1)}$. $E[X_i]=\frac{1}{n(n-1)}\Sigma_{s,t\in V}d(s,t)$. We sample pairs independently with replacement. So, the variables, $X_i$'s are i.i.d. Now,  $E[\frac{1}{p}\sum\limits_{i=1}^pX_i]=
\frac{1}{p}\sum\limits_{i=1}^pE[X_i]=\frac{1}{p} \cdot \frac{p}{n(n-1)} \cdot \Sigma_{s,t\in V}d(s,t)=\mu$.
\end{proof}

We employ Hoeffding's inequality~\cite{hoeff1963} to bound the error produced by our sampling method in a single greedy step. Hoeffiding's inequality provides a sample-size dependent bound for the difference between the estimated mean (based on samples) and the actual mean of a population. 
The requirement for the applicability of Hoeffiding's inequality is that the summed variables are chosen independently from the same distribution, which is the case in our setting. 
Similar independent node pair sampling analysis using Hoeffding's inequality has been previously employed by Yoshida et al.~\cite{yoshida2014} 
to estimate the group betweenness of vertices, whereas, we estimate the reduction in the sum of shortest paths upon node upgrades.  
In what follows, we demonstrate that the estimate has low error with high probability requiring only small number of samples. Furthermore, we show the same quality guarantee with even smaller number of samples in small-world networks.


\begin{thm}\label{thm:approxgeneral}
Given a target set $S$ and a sample $P$ of size $p$, if $v_g$ and $v_a$ are the next vertices chosen by GR and GS respectively, the difference in delay reduction due to these choices is bounded as follows:
\begin{equation*}
Pr[|RR(v_g|S) - RR(v_a|S)| < \epsilon]>1-\frac{1}{n^2},
\end{equation*}
where $p$ is $O(\frac{c^2logn}{\epsilon^2})$, $c=\frac{diam}{l_{min}}$, $diam$ and $l_{min}$ are the diameter and minimum delay respectively. 
\end{thm}

\begin{proof}
Let $M_g=\Sigma_{s,t\in V}d(s,t|S\cup\{v_g\})$ and $M_a=\Sigma_{s,t\in V}d(s,t|S\cup\{v_a\})$. Let also $\mu_g$ and $\mu_a$ denote the corresponding mean distances and $Y^{g}$ and $Y^{a}$ be the corresponding expected means computed using the samples. 

Since the samples provide an unbiased estimate (Lemma~\ref{lemma:means_general}) and are i.i.d., we can use Hoeffding's inequality~\cite{hoeff1963} to bound the error of the mean estimates:
\begin{equation*}
Pr[|Y^{g}-\mu_g| \geq \beta]\leq \delta 
\end{equation*}
where $\delta=2\exp(-\frac{2p^2\beta^2}{\Lambda})$, $X_i$ represents the distance between the $i-$th pair of vertices in the sample, $a_i\leq X_i \leq b_i$, and $\Lambda =\sum\limits_{i=1}^p(b_i-a_i)^2$. Similarly,
$Pr[|Y^{a}-\mu_a| \geq \beta]\leq \delta.$


Applying union bound, $Pr[(|Y^{g}-\mu_g| \geq \beta) \cup (|Y^{a}-\mu_a| \geq \beta)] \leq 2\delta.$ By construction,  $\mu_g \geq \mu_a$ as GR selects the best next vertex at each step. On the other hand, since GS selects $v_a$, it must be that $Y^{a} \geq Y^{g}$. As, the sampled best node is probabilistic, we need to apply union bound over $n$ possible nodes.
As a consequence, we get 
$Pr[|\mu_g - \mu_a| \geq 2\beta] \leq 2n\delta$, or alternatively 
$Pr[|\mu_g - \mu_a| < 2\beta] > 1 - 2n\delta$. 
\newline

Now, $Pr[|RR(v_g|S)-RR(v_a|S)| < \epsilon]  $\newline 
$~~~~~=Pr[|M_a-M_g| < \epsilon.SPD]$\newline
$~~~~~=Pr[|\mu_a-\mu_g| < \frac{\epsilon.SPD}{n(n-1)}]$
$> 1- 4n\exp(-\frac{2p^2(\frac{\epsilon.SPD}{2n(n-1)})^2}{\Lambda})$\newline

 But, $1- 4n\exp(-\frac{2p^2(\frac{\epsilon.SPD}{2n(n-1)})^2}{\Lambda})>1-4n\exp(-\frac{p(\epsilon.l_{min})^2}{2diam^2})$, 
since $SPD > n(n-1)l_{min}$ and $(b_i-a_i) < diam$ and as a consequence, $Pr[|RR(v_g|S)-RR(v_a|S)| < \epsilon] > 1-4n\exp(-\frac{p(\epsilon.l_{min})^2}{2diam^2})$.\newline

Thus, by choosing $p=\frac{2diam^2log(4n^3)}{(\epsilon.l_{min})^2}$, we have \newline
$~~~~~~~~~~~~~ Pr[|RR(v_g|S)-RR(v_a|S)| < \epsilon] >1-\frac{1}{n^2}$
\end{proof}

Note that, in the theorem, we assume $l_{min} >0$ without loss of generality. If $l_{min}=0$, one can delete any node of zero delay, add all possible edges among its neighbours and consider the resulting network as an input. For small-world networks (where the diameter is $\leq l_{max}logn$), a property exhibited in many domains, we show that the number of samples needed to obtain the same quality is much smaller. 
\begin{cor}
Given a small-world network in which $diam\leq l_{max}logn$, the error of GS using $p=O(\frac{log^3n}{\epsilon^2})$ samples  can be bounded as:\newline
\begin{equation*}
Pr[|RR(v_g|S)-RR(v_a|S)| < \epsilon]>1-\frac{1}{n^2}
\end{equation*}
\end{cor}
\begin{proof} Let the network have a small-world property, $diam\leq l_{max}logn$, where $l_{max}$ is the maximum delay. Now $c$ in above theorem can be replaced by $\frac{l_{max}logn}{l_{min}}$, where $l_{min}$ is the minimum delay. $\frac{l_{max}}{l_{min}}$ is assumed to be constant.
\vspace{-1mm}
\end{proof}

\begin{algorithm}[t]
\scriptsize
 \caption {Greedy with Sampling (GS)\label{algo:GS}}
\begin{algorithmic}[1] 
 \REQUIRE Network $G=(V,E,L)$, Approximation error $\epsilon$, Sampling factor $c$, Budget $k$
 \ENSURE A subset of $k$ nodes, Target Set 
\STATE Choose $p =  O(clogn/\epsilon^{2})$ pairs of vertices in $P$
\STATE $T\leftarrow\Phi$
\WHILE {  $|T|\leq k$ }
\FOR {  $(s,t)\in P$ }
\STATE Compute $d(s,t'|T)$ and $s.target[t']\leftarrow d(s,t'|T) \hspace{.1cm} \forall t'\in V$
\STATE Compute $d(s',t|T)$ and $t.source[s']\leftarrow d(s',t|T) \hspace{.1cm} \forall s'\in V$
\ENDFOR

\FOR{$v' \in V$}
\IF {$l(v')>0$}
\STATE $R_{v'}\leftarrow \Sigma_{(s,t)\in P}d(s,t|T)-\Sigma_{(s,t)\in P}d(s,t|T\cup\{v'\})$
\ENDIF
\ENDFOR
\STATE $v\leftarrow max_{v'\in V}\{R_{v'}\}$
\STATE $l(v)\leftarrow 0$ and $T\leftarrow T\cup \{v\}$
\ENDWHILE
\STATE Return $T$
\end{algorithmic}
\end{algorithm}
\vspace{-1mm}

GS (Alg.~\ref{algo:GS}) takes as input a network $G$, a target approximation error $\epsilon$, a sampling factor $c$ and a budget $k$. The algorithm outputs a target set of vertices constructed based on optimizing the sum of the distances between each of the sampled pair paths. The approximation error, $\epsilon$, defines the difference between the approximate and the optimal reduction at each step. The number of samples $p$ depends on the number of vertices $n$, the error $\epsilon$, and the sampling factor $c$. In theory $c$ should be chosen as shown in the theorem based on the input graph $G$. But in practice, we use a small constant $c$, requiring small number of samples (see Sec. \ref{sec:exp}). 

The algorithm first samples pairs from the population of all pairs (step 1). Note that, although we present the algorithm with sampling pairs once, sampling new $p$ pairs in each iteration does not change the quality bounds or running time of the algorithm. The algorithm runs for $k$ iterations. It computes the sum of  distances between each of the sampled pairs and selects the best vertex which reduces this sum the most. 
 To achieve this, in each iteration, we computes the desired shortest path distances and store them (step 5-6 ). Next we select the vertex with maximum reduction in the sum of distances of the sampled pairs (step 8-12). The selection is conditioned on already selected vertices in TS. 

\noindent \textbf{Complexity:} 
The running time of GS is dominated by the computation of shortest paths. If the number of edges is $O(nlogn)$, running Dijkstra's algorithm for each sampled pair takes $O(nlogn)$. The algorithm has $k$ iterations. This leads to the time complexity of $O(kpnlogn)$. We need to store only the distances from the end vertices of the pair to all other vertices. This leads to a space complexity of $O(pn)$.

\subsection{Uniform Model: Approximate Target Set}

In some applications, instances of our design problem may feature uniform (equal) or close-to-uniform initial delays. For example, many routing devices in a computer network might have similar hardware configuration and hence feature comparable delays. Similarly, intersections with the same number of lanes within a road network allow for similar rate of cars to propagate during congestion periods. Such homogeneous instances offer more structure to the design problem and allow for a better (faster and higher-quality) sampling scheme than our general-case algorithm GS.     
Hence, we develop and analyze a superior sampling based-method, called PCS (Path Count with Sampling), targeted for the uniform model. 
 
 
We relate the delay reduction due to a vertex to the number of shortest paths passing through it. Let $\zeta^v(S)$ ( or $\zeta^v$, we are omitting $S$ for simplicity) denote the number of shortest paths passing through a vertex $v$ assuming that $S$ is the target set.
\begin{thm}\label{thm:pathcount}
In the uniform model, for a given set $S$ and $v\notin S$,
\label{eqn:pc}
$RS(v|S)=\zeta^v+(n-1)$.
\end{thm}

\begin{proof}
See the Appendix.
\end{proof}

With the above result, a greedy algorithm only needs to know the values of $\zeta$ for each vertex. The main bottleneck of computing $\zeta$ involves shortest path computation between all pairs of vertices. We address this complexity by a different sampling scheme. We estimate $\zeta$ for a vertex based on the shortest paths among $p$ pairs of vertices sampled independently with replacement. 
Let $X^{v}$ be a random variable denoting the number of times $v$ belongs to $SP_{s,t}$ for all sampled pairs $(s,t)$, where $SP_{s,t}$ ($s,t \notin SP_{s,t}$) denotes the set of vertices on the shortest path(s) between $s$ and $t$. The expected value of the random variable is computed as follows:

\begin{lemma} \label{lemma:uniform_lemma}
\vspace{-1mm}
For any vertex $v$,
$E[X^{v}]= \frac{p}{n(n-1)}\zeta^v.$
\vspace{-1mm}
\end{lemma}

The lemma holds due to the additive property of expectation and the fact that the pairs are sampled independently. Next, we show that the difference in quality of GR and PCS is small with high probability in a single greedy step.

\begin{thm}\label{thm:approxpathcount}
Given a sample $P,|P|=p=O(\frac{logn}{\epsilon^2})$, if $v_g$ and $v_a$ are the vertices chosen by GR and PCS respectively, then
\begin{equation*}
Pr[|RR(v_g|S) - RR(v_a|S)| < \epsilon] > 1-\frac{1}{n}.
\end{equation*}
\end{thm}
\begin{proof}
If $X_1,X_2,...,X_p$ are independent random variables in $[0,1]$ and $\bar{X}=\frac{1}{p}\sum\limits_{i=1}^pX_i$, then from Hoeffding's inequality~\cite{hoeff1963}:
$Pr[|\bar{X}-E[\bar{X}]|\geq \beta ] \leq 2\exp(-2\beta^2p).$

Using Lemma \ref{lemma:uniform_lemma}
and Hoeffding's inequality, \newline
$Pr[|\frac{\zeta^{v_g}}{N}-\frac{1}{p}X^{v_g}| \geq \beta ] \leq 2\exp(-2\beta^2p)$, and similarly\newline
$Pr[|\frac{\zeta^{v_a}}{N}-\frac{1}{p}X^{v_a}| \geq \beta ] \leq 2\exp(-2\beta^2p)$. The optimal vertex chosen by PCS is $v_a$ and hence $X^{v_a} \geq X^{v_g}$. Since $X^{v_a} \geq X^{v_g}$, and $\frac{\zeta^{v_g}}{N}\geq\frac{\zeta^{v_a}}{N}$ (by construction), we apply the same logic of union bound as in Theorem \ref{thm:approxgeneral} to achieve
$Pr[|\frac{\zeta^{v_g}}{N}-\frac{\zeta^{v_a}}{N}| < 2\beta ] > 1-4n\exp(-2\beta^2p)$. \newline

Now, we use this inequality to derive the following:\newline
$~~~~~~~~Pr[|RR(v_g|S)-RR(v_a|S)| < \epsilon]  $ \newline
\vspace{1mm}
$~~~~~~=Pr[|RS(v_g|S)-RS(v_a|S)| < \epsilon.SPD] $ \newline
\vspace{1mm}
$~~~~~~=Pr[|\zeta^{v_g}-\zeta^{v_a}| < \epsilon.SPD]$, from Thm.~\ref{thm:pathcount}\newline
\vspace{1mm}
$~~~~~~=Pr[|\frac{\zeta^{v_g}}{N}-\frac{\zeta^{v_a}}{N}| < \frac{\epsilon.SPD}{N}]$ \newline
$~~~~~~~~>1- 4n\exp(-2p(\frac{\epsilon.SPD}{2N})^2)$\newline
$~~~~~~~~>1-4n\exp(-\frac{p\epsilon^2}{2})$, Since $SPD > N$.\newline

If we choose $ p= \frac{2log(4n^3)}{\epsilon^2}$, then \newline
$~~~~~~~~Pr[|RR(v_g|S)-RR(v_a|S)| < \epsilon] >1-\frac{1}{n^2}$.
\end{proof}

Thm. \ref{thm:approxpathcount} shows that the error of PCS w.r.t. GR is bounded by $\epsilon$ with probability $1-\frac{1}{n^2}$ at a single step. The number of samples needed by PCS is $O(\frac{log(n)}{\epsilon})$; this is a factor of $O(log^2n)$ less than the number of samples needed in GS for small-world networks and a factor of $O(\frac{diam^2}{l_{min}^2})$ less in general networks. 

\begin{algorithm}[t]
\scriptsize
 \caption {Path Count with Sampling (PCS)\label{algo:PCS}}
\begin{algorithmic} [1]
 \REQUIRE $G=(V,E,L)$, Approximation error $\epsilon$, Budget $k$
 \ENSURE A subset TS of $k$ nodes 
\WHILE {  $|T| \leq k$ }
\STATE Choose $p =  O(logn/\epsilon^{2})$ pairs of vertices in $P$
\FOR {  $(s,t)\in P$ }
\IF {$l(s)=1$ and $l(t)=1$}
\STATE Perform BFS from $s$ to $t$ and add vertices to $SP_{s,t}$
\ENDIF
\IF {($l(s)=1$ and $l(t)=0$) or ($l(s)=0$ and $l(t)=1$)}
\STATE If ($l(s) = 0$ and $l(t) = 1$) then swap $s$ and $t$
\STATE Perform BFS from $s$ to $\{u|u\in t.gateway\} $
\STATE Add vertices of $SP_{s,u}$ to $SP_{s,t}$; $u\in \{u'|d(s,u')\leq d(s,u_1)\ \forall u_1\in t.gateway\}$
\ENDIF

\IF {$l(s)=0$ and $l(t)=0$}
\STATE Perform BFS from $u\in s.gateway$ to $u'\in t.gateway$
\STATE Add $v_1, v_2$ and vertices of $SP_{v_1,v_2}$ to $SP_{s,t}$; $(v_1,v_2)\in \{(u,u')|d(u,u')\leq d(u_1,u_2),\ \forall u,u_1\in s.gateway,\ \forall u',u_2\in t.gateway\}$
\ENDIF
\STATE $\zeta_{v'}\leftarrow \zeta_{v'}+1$ if $v'\in SP_{s,t}$
\ENDFOR
\STATE $v\leftarrow max_{v'\in V}\{\zeta_{v'}\}$
\STATE $l(v)\leftarrow 0$ and $T\leftarrow T\cup \{v\}$
\STATE Edit $G$: Delete vertex $v$ and add edges (if absent) between\\ its neighbors
\STATE Update the $gateway$ list for each $u\in T$ if necessary
\ENDWHILE
\STATE Return $T$
%
\end{algorithmic}
\end{algorithm}

Algorithm \ref{algo:PCS} (PCS) computes TS based on estimates of number of shortest paths through each vertex. The approximation error $\epsilon$ bounds the difference between reduction by PCS and Greedy (GR) in each iteration. In each of the $k$ iterations, PCS first samples $p$ pairs of nodes from the population of all pairs (the proven approximation still holds when the samples are obtained before the iteration starts as in Alg. \ref{algo:GS}). The overall complexity is $O(kp(m+n))$ (for details of the algorithm and its running time see the Appendix).



\vspace{-1mm}
\section{Experimental Results}
\label{sec:exp}
We evaluate the quality and scalability of our algorithms in both synthetic and real-world networks. We conduct all experiments on $3.30$GHz Intel cores with $30$ GB RAM. All algorithms are implemented in Java. 

\subsection{Datasets}

The real-world datasets for evaluation are listed in Table~\ref{table:data_description}. The \emph{air transportation} (\url{http://www.rita.dot.gov}) data consist of airline flight networks with delays at airports set according to historical flight delays due to circumstances within the airline's control (e.g. maintenance or crew problems, aircraft cleaning, baggage loading, fueling, etc.). We consider average and total delay of flights originating from an airport in the period 01/13-09/15. Our \emph{Traffic} data is from the highway network of Los Angeles, CA~\cite{silva2015}, where the delay at an intersection is defined as the scaled inverse of the observed speed at a given point in time ($1500* 1/speed$). According to this definition the delay values range between $15$ and $80$ (similar to that of the original speeds). 
The Twitter dataset is a social network in which edges correspond to follower relationships among users. We disregard the direction of edges for our analysis. Node delays in this network represent the average inter-arrival time between posts on a given topic. We experiment with different topics described in~\cite{bogdanov2013}. The vertices in the \emph{DBLP} network are authors and the edges represent co-authorship on at least one paper. For DBLP, we assign delays randomly, with values uniformly distributed in multiples of ten between $10$ to $100$. Our goal is to evaluate the scalability of our algorithms on a large real-world network structure.


\begin{table}[h]
\centering
\scriptsize
\begin{tabular}{| c | c | c | c | c | c |}
\hline
\textbf{name}& \textbf{value}& \textbf{$|V|$} & \textbf{$|E|$}\\
\hline
\textbf{Jetblue (JB)}& carrier delay & 63 & 172\\
\hline
\textbf{Southwest (SA)}& carrier delay & 89 & 716\\
\hline
\textbf{American (AA)}& carrier delay & 100 & 363\\
\hline
\textbf{Delta (DA)}& carrier delay & 160 & 553\\
\hline
\textbf{Traffic}& inverse speed & 2K & 6K\\
\hline
\textbf{Twitter-Celeb} & posting delay &  28K & 240K\\
\hline
\textbf{Twitter-Politics} & posting delay &  100K & 7.4M\\
\hline
\textbf{Twitter-Science} & posting delay &  100K & 3.3M\\
\hline
\textbf{DBLP}& random & 1.1M  & 5M\\
\hline
\end{tabular}
\vspace{-1mm}
\caption{\small Dataset description and statistics. \label{table:data_description}}
\vsa
\vspace{-2mm}
 \end{table}



\begin{figure}[t]
\vsa
    \centering
    \vsc
    \subfloat[Traffic-Uniform]{\includegraphics[width=0.23\textwidth]{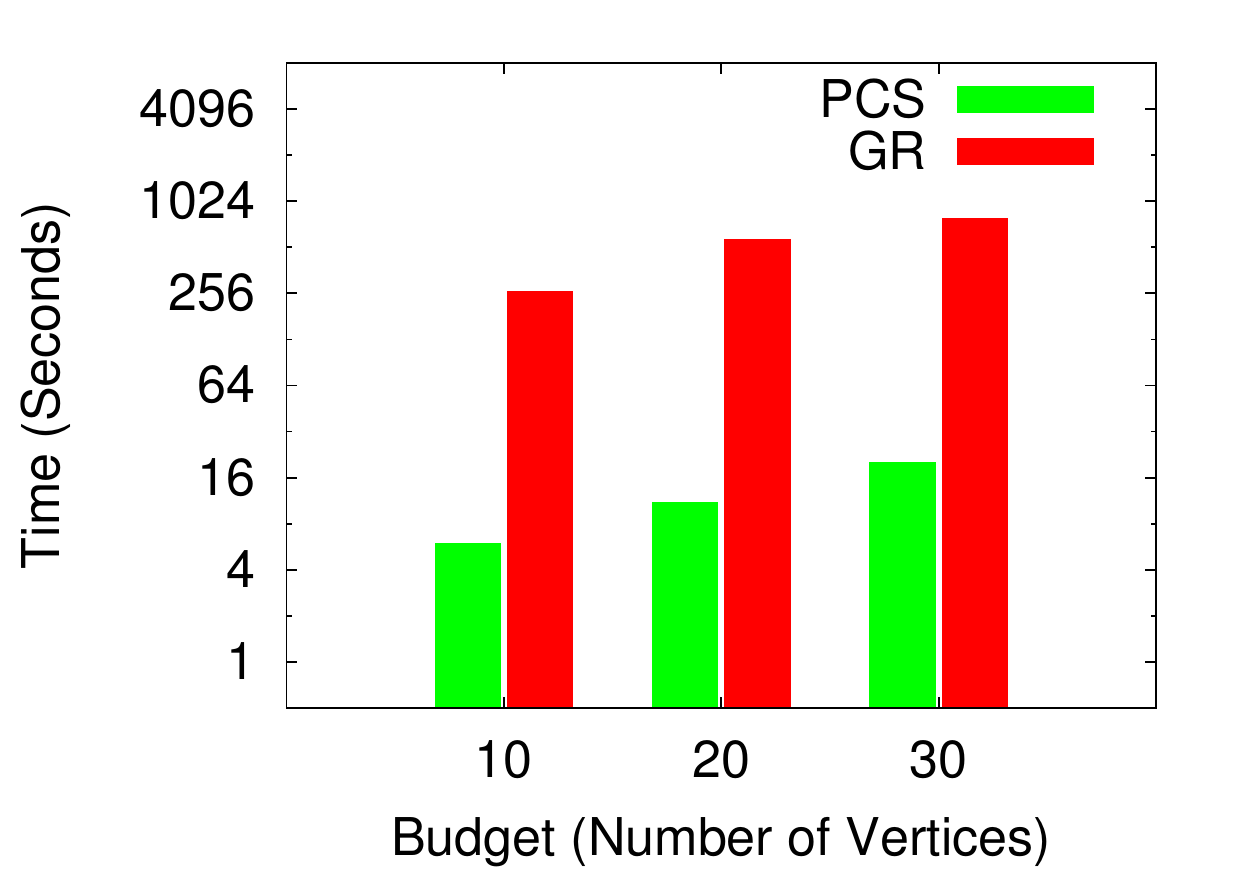}\label{fig:PCSGR_effi_traffic}}
    \subfloat[Traffic-Uniform]{\includegraphics[width=0.23\textwidth]{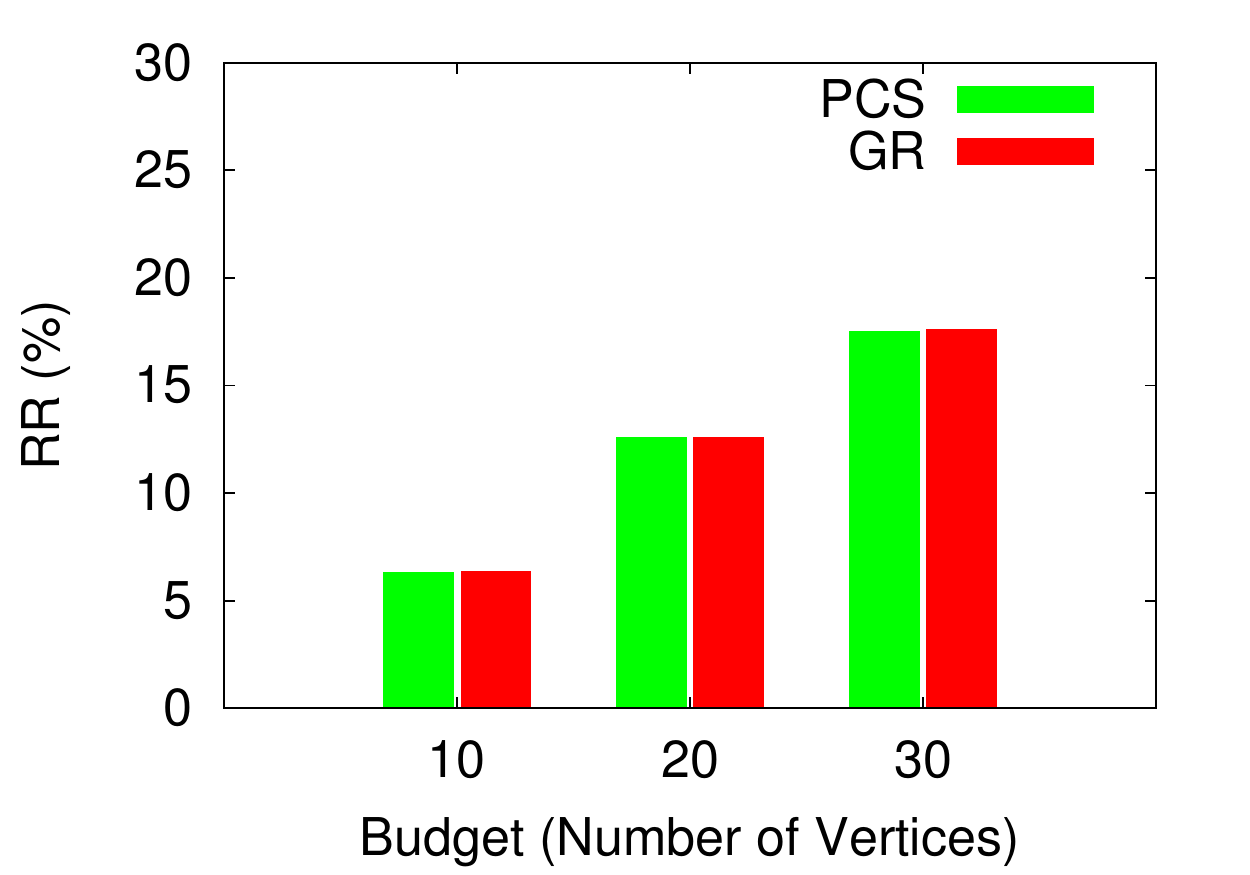}\label{fig:PCSGR_qual_traffic}}\\ \vsc
    \subfloat[Traffic-Delay]{\includegraphics[width=0.23\textwidth]{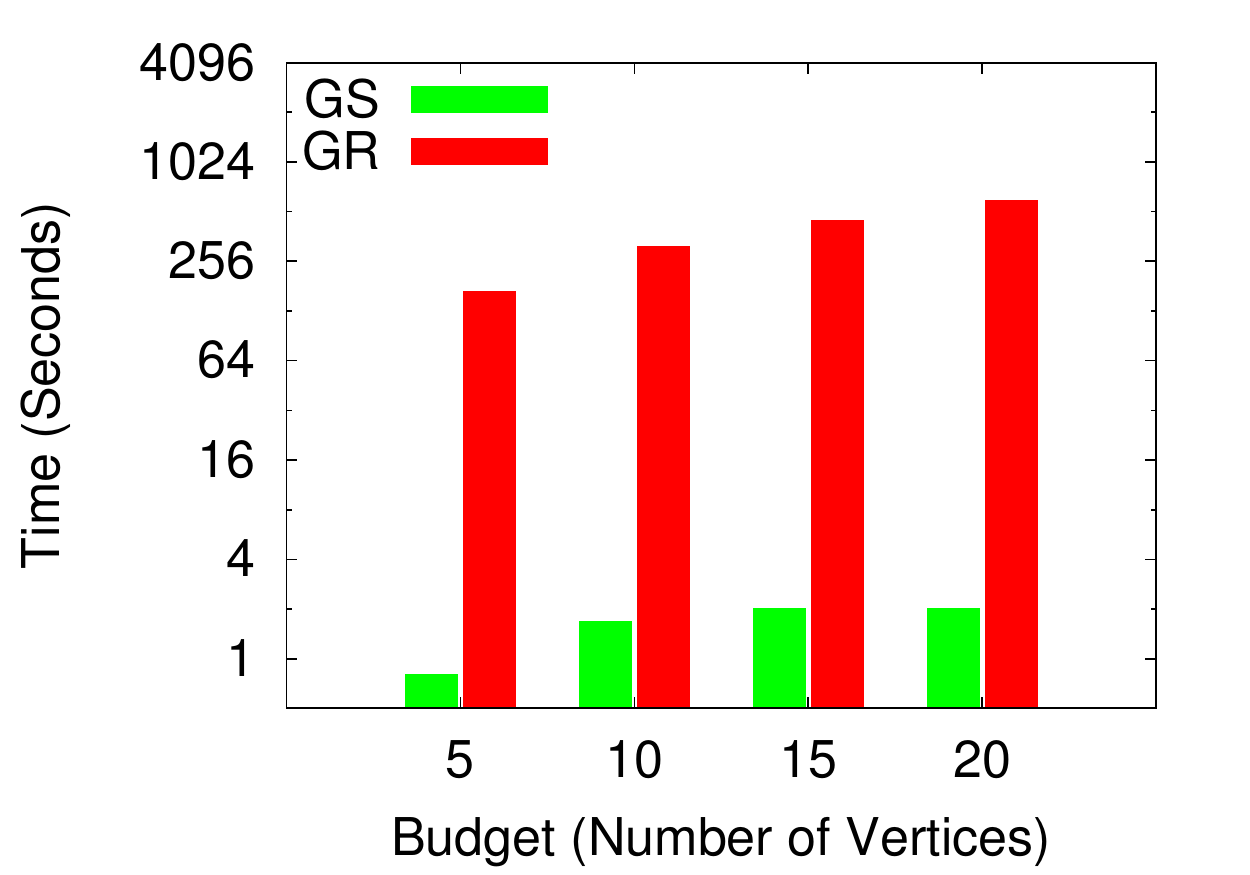}\label{fig:GSGR_effi_traffic_sp1}}
    \subfloat[Traffic-Delay]{\includegraphics[width=0.23\textwidth]{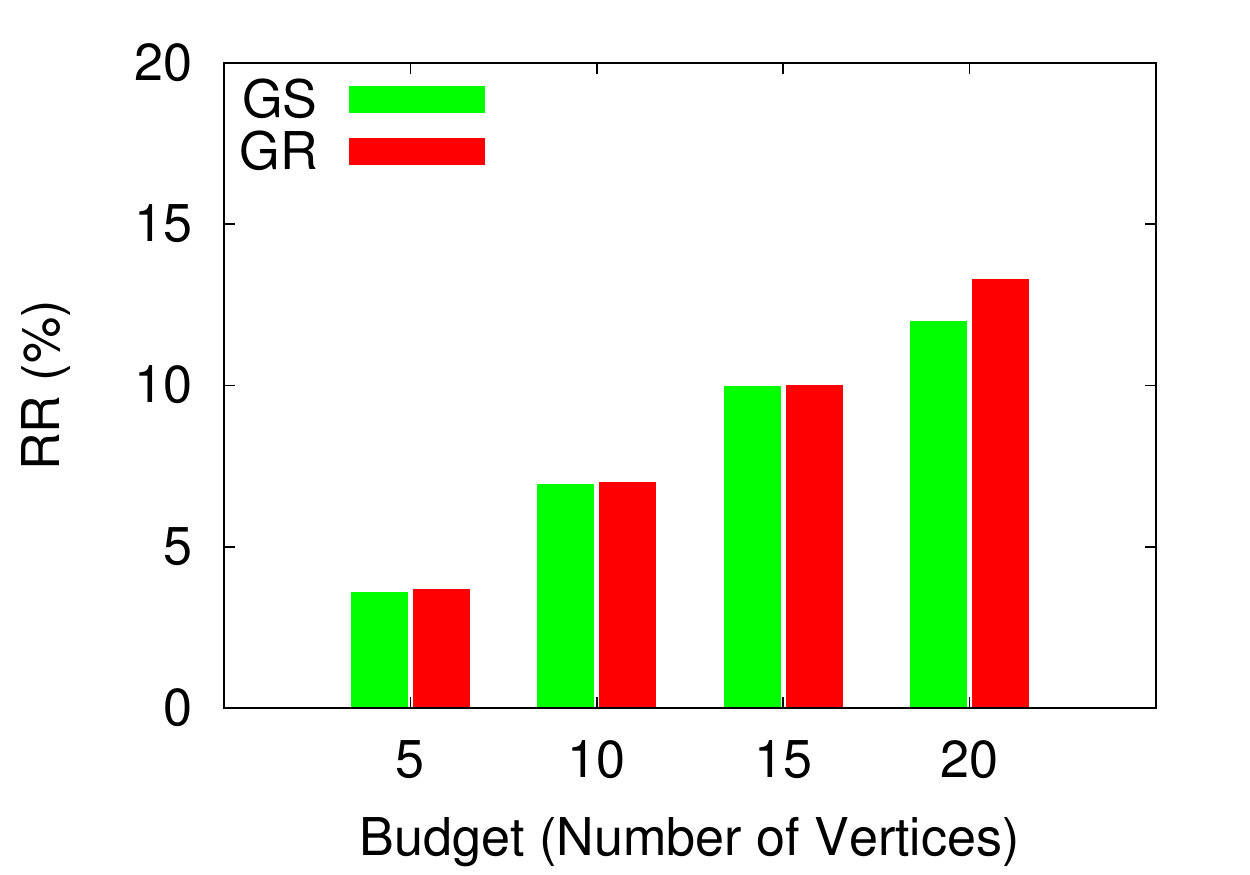}\label{fig:GSGR_qual_traffic_sp1}}
\vspace{-2mm}
    \caption{\small Uniform model: (a-b) Execution time and relative reduction of Greedy (GR) and Path Counting (PCS); General model: (c-d) Greedy (GR) and sampling-based Greedy (GS) for Traffic.\label{fig:GSGR_comparison}}
    \vspace{-1mm}
\vsb \end{figure}

\subsection{Quality of sampling compared to Greedy}
We report the number of samples in PCS and GS as $c*logn$, where $c$ is related to the expected error $\epsilon$ in Thms.~\ref{thm:approxgeneral} and \ref{thm:approxpathcount}. Unless stated otherwise, we use $c=10$. 

First, we compare our sampling schemes GS and PCS with Greedy (GR) in order to evaluate the effect of sampling on quality, which we theoretically analyze in Thms. \ref{thm:approxgeneral}, \ref{thm:approxpathcount}. To enable the comparison, we use small real datasets due to the limited scalability of GR. 
The quality of the compared algorithms is quantified as the Relative Reduction (RR) of SPD, while efficiency---in terms of wall-clock time. We use $3.5logn$ samples for GS and PCS in these experiments.

In all experiments, our sampling schemes achieve similar quality as that of Greedy (GR), while taking close to two orders of magnitude less time. In the uniform model, the difference in quality between our sampling scheme PCS and GR does not exceed $.05\%$ in the traffic dataset (Fig.~\ref{fig:PCSGR_qual_traffic}), while PCS takes only $2\%$ of the time taken by GR (Fig.~\ref{fig:PCSGR_effi_traffic}). 
This trend persists in the case of the general delay model for which we employ our sampling-based Greedy (GS). 
We compare GS and GR on multiple snapshots of the Traffic dataset and report average completion times and quality in Fig. \ref{fig:GSGR_effi_traffic_sp1},\ref{fig:GSGR_qual_traffic_sp1}. GS is $200$ times faster than GR and its solution's RR is only $0.1\%$ worse than that of GR Figs.~\ref{fig:GSGR_effi_traffic_sp1},\ref{fig:GSGR_qual_traffic_sp1}. We observe similar results in synthetic data (see the Appendix).


\begin{table}[t]
\vspace{-1mm}
\centering
\scriptsize
\begin{tabular}{| c | c | c | c | c | c |}

\hline
&\multicolumn{3}{c|}{\textbf{Quality RR(\%)}} & \multicolumn{2}{c|}{\textbf{Time [sec.]}}\\
\hline

\textbf{Airlines}& GR& GS& High-Delay & GS & GR\\
\hline
\textbf{JB}& $68.3$ & $68.2$ &  $63.9$ & $0.1$ & $0.22$\\
\hline
\textbf{SA}&$58.5$ &  $58.3$ &  $58.5$ & $0.4$ & $0.53$\\
\hline
\textbf{AA}& $55$ & $54.3$ &  $5.45$ & $0.06$ & $0.9$\\
\hline
\textbf{DA}& $48.9$ & $48.4$ &  $4.44$ & $0.4$ & $0.8$\\
\hline
\end{tabular}
\vspace{-1mm}
\caption{ \small Comparison on the airlines dataset. For budget $5$, columns 2-4 show the RR for GR, GS, and High-Delay respectively and columns 5-6 show running times. \label{table:results_airlines}}
\vspace{-3mm}
 \end{table}

In the airline data we assign node delay as the average airline-induced delay of all historical flights originating from an airport. Table~\ref{table:results_airlines} summarizes our results. GS matches the quality of GR in a fraction of the computation time. For AA, our solution selects important nodes which are central and also have significant delays. The solution contains hub airports like those in Dallas, Charlotte and Phoenix, since improving these airports makes them more central in the network and collectively improves the total end-to-end delay by $55\%$~\ref{fig:AA}. Surprisingly, considering only node delay (baseline "High-Delay") has significant disadvantages. Some non-central airports in the AA and DA networks have significant average delays and hence disregarding the network position results in a $10$-fold worse quality of the High-Delay baseline. All other baselines do not exceed the quality of GS for varying budgets and airlines.

\begin{table}[t]
\centering
\scriptsize
\begin{tabular}{| p{1.2cm} | p{1.5cm} |p{1.8cm}| p{2.3cm} |}
\hline
\textbf{Algs.}& \textbf{Selection} & \textbf{Uniform} & \textbf{General}\\
\hline
Random & $10$ trials  & O(n) & O(n) \\\hline
Deg-Cen & Degree & O(nlogn) & O(nlogn) \\ \hline
High-Delay & Delay & O(nlogn)) &O(nlogn) \\\hline
Path-Cen\cite{dilkina2011} & Delay$\times$\#SPs & O(n(m+n)) & O(n(m+nlogn))\\\hline
It-Path-Cen\cite{dilkina2011} &  Delay$\times$\#SPs w. updates  & O(kn(m+n)) & O(kn(m+nlogn)) \\\hline
PCS & Alg.~\ref{algo:PCS} & O(kc(m+n)logn) & \\\hline
GS & Alg.~\ref{algo:GS}  &  & O(kc(m+nlogn)logn) \\\hline
GR & Alg.~\ref{algo:GR}  & O(k$n^3$) & O(k$n^3$) \\\hline
\end{tabular}
\vspace{-1mm}
\caption{\small Theoretical complexity of compared algorithms.}\label{tab:baselines}
\vsc \end{table}


\begin{figure*}[t]
    \centering
\subfloat[Twitter-Celeb]{\includegraphics[width=0.2\textwidth]{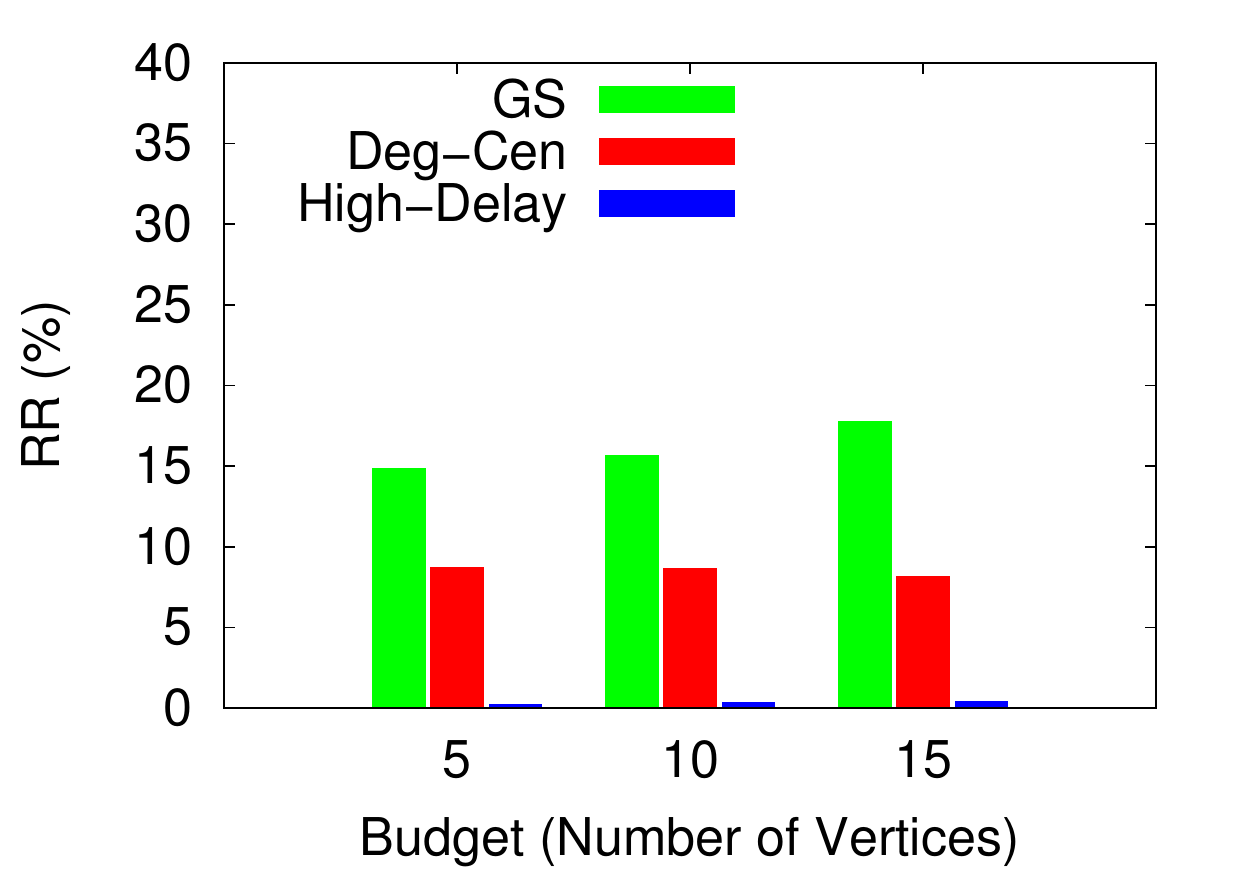}\label{fig:twitter_celeb_qual}}
\subfloat[Twitter-Politics]{\includegraphics[width=0.2\textwidth]{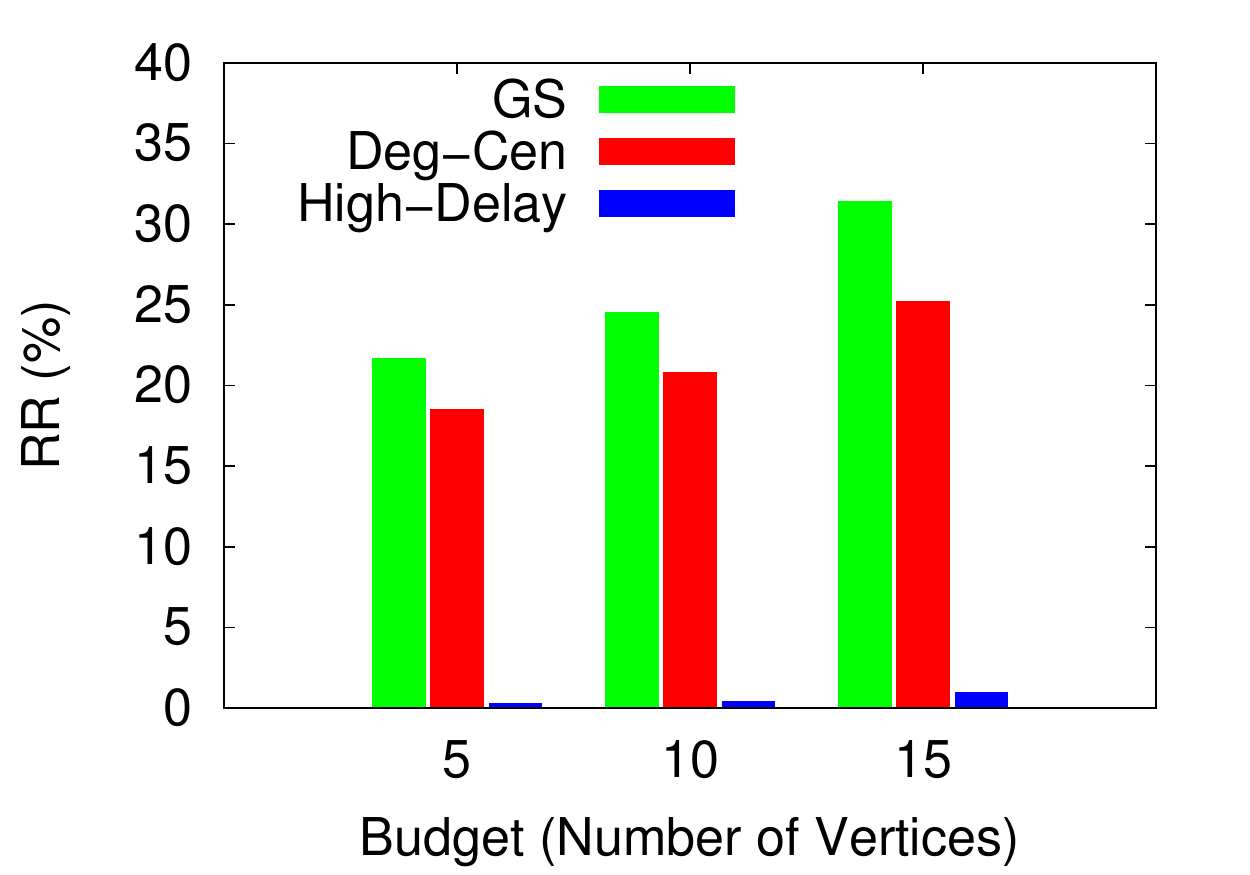}\label{fig:twitter_politics_qual}}
\subfloat[Twitter-Science]{\includegraphics[width=0.2\textwidth]{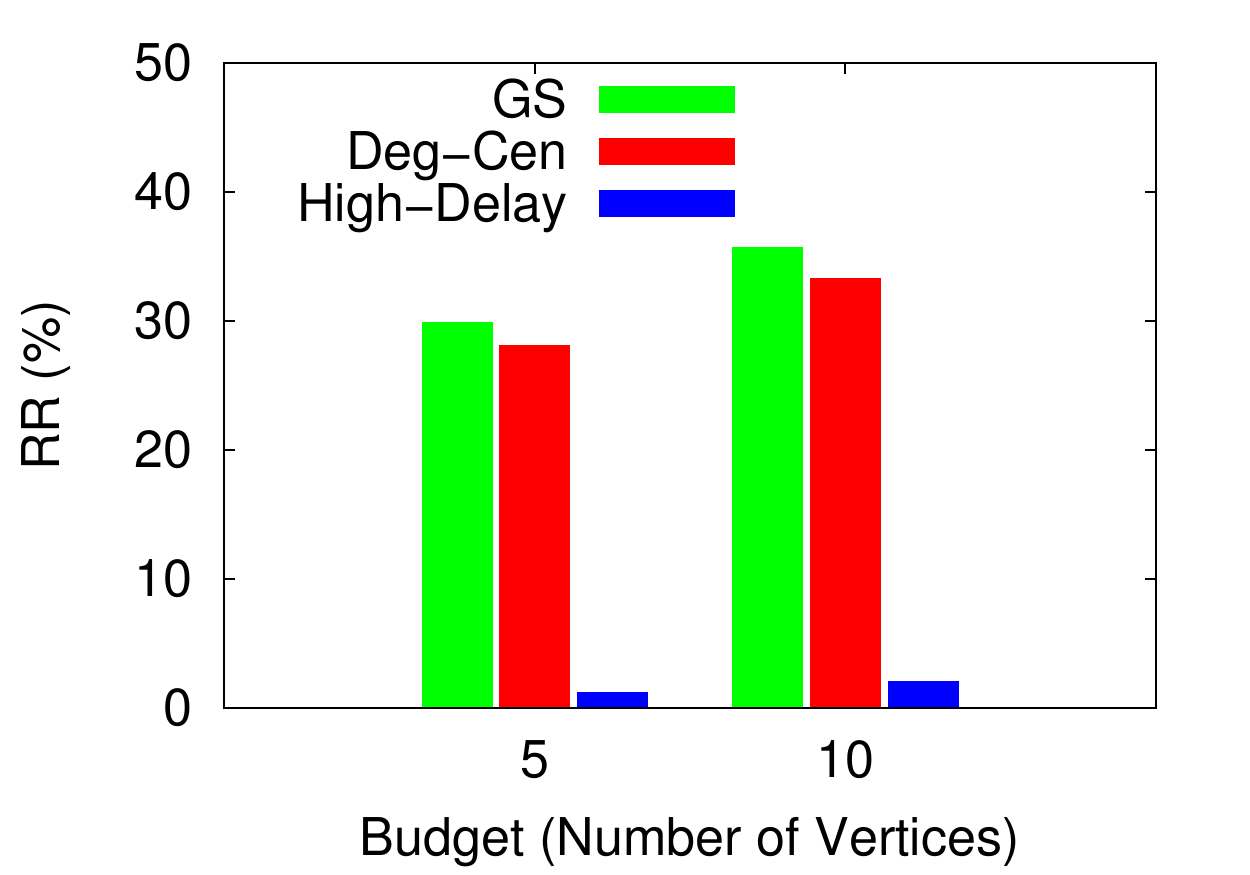}\label{fig:twitter_science_qual}}
\subfloat[DBLP-Random]{\includegraphics[width=0.2\textwidth]{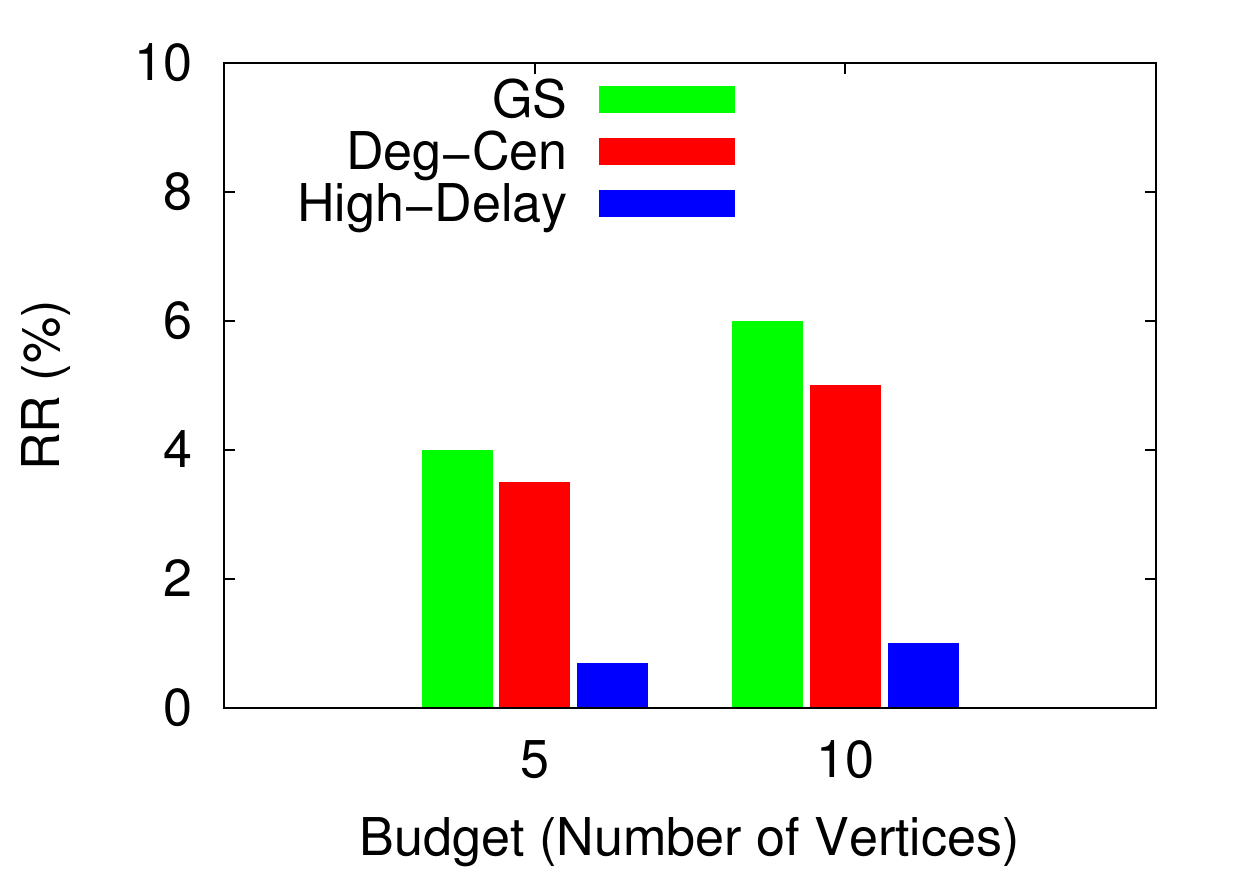}\label{fig:dblp_random_qual}}
\subfloat[Uniform]{\includegraphics[width=0.2\textwidth]{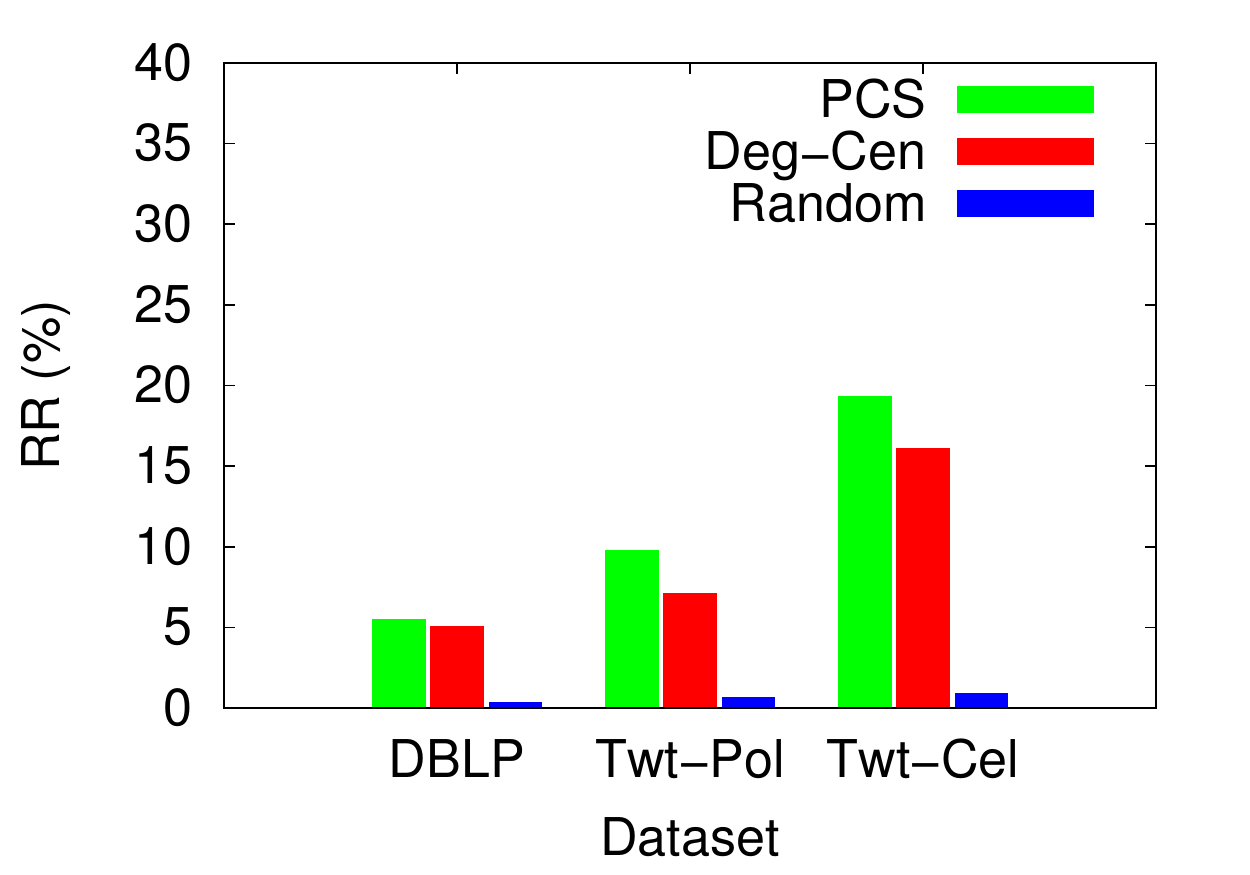}\label{fig:PCS_budget5}}
\vspace{-2mm}
   \caption{ \small (a-d) General Model:  Quality of GS and baselines on Twitter-Celeb, Twitter-Politics, Twitter-Science and DBLP. (e) Uniform Model: Quality of PCS and baselines for budget=5 on DBLP, Twitter-Politics and Twitter-Celeb.\label{fig:large}}
\vsc \end{figure*}


\vspace{-1mm}
\subsection{Comparison to baselines}\label{sec:baselines}
Next we evaluate the performance of our algorithms in comparison to alternatives. We consider several baseline methods, listed in Tab.~\ref{tab:baselines} along with their theoretical running times and those of our algorithms. Some baselines select TS vertices based on local properties: degree (Deg-Cen) or delay (High-Delay); while others---based on the product of global path centrality and delay (Path-Cen and It-Path-Cen~\cite{dilkina2011}). It-Path-Cen updates the number of shortest paths through a vertex after each selection of a target vertex. 

\begin{figure}[h]
\vspace{-5mm}
    \centering
    \subfloat[Traffic-Uniform]{\includegraphics[width=0.22\textwidth]{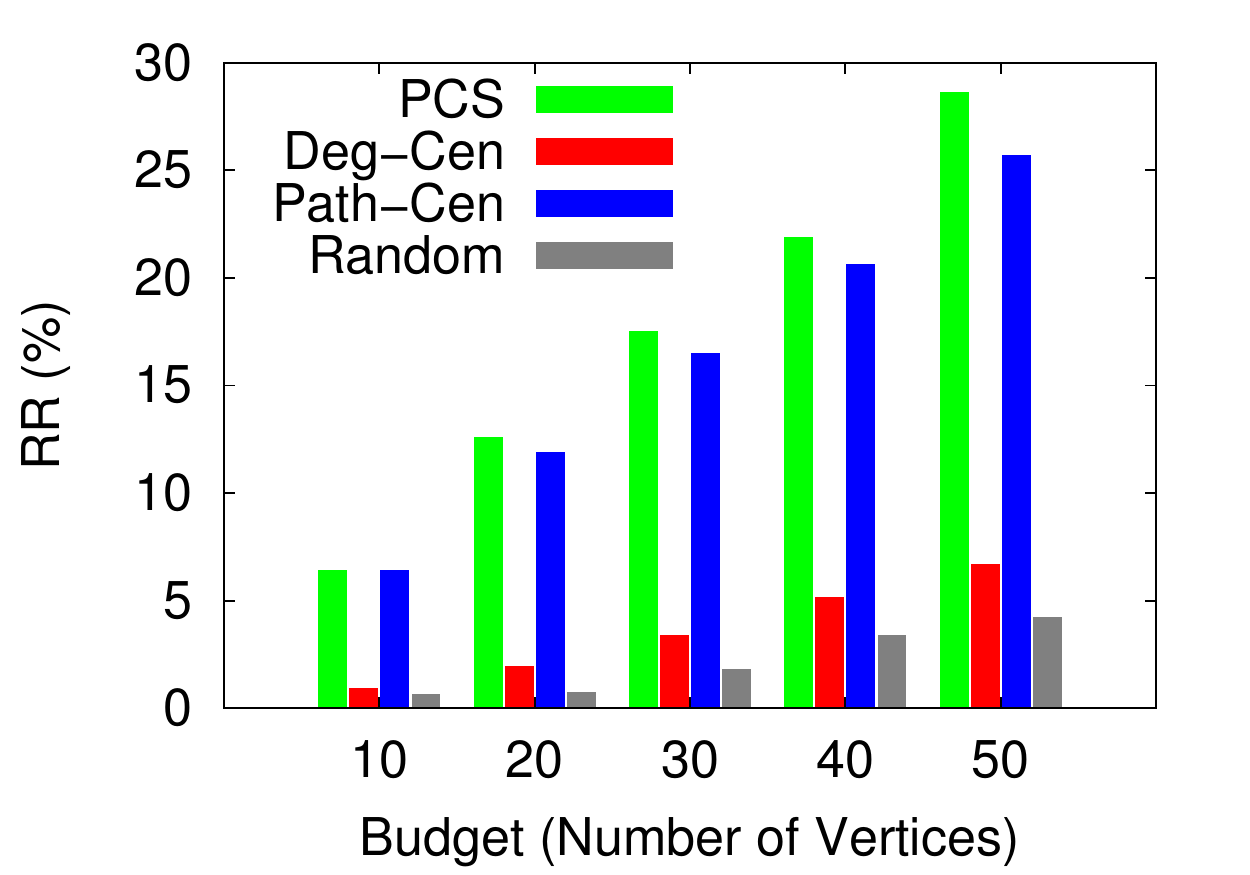}\label{fig:PCS_qual_traffic}}
    \subfloat[Traffic-Delay]{\includegraphics[width=0.22\textwidth]{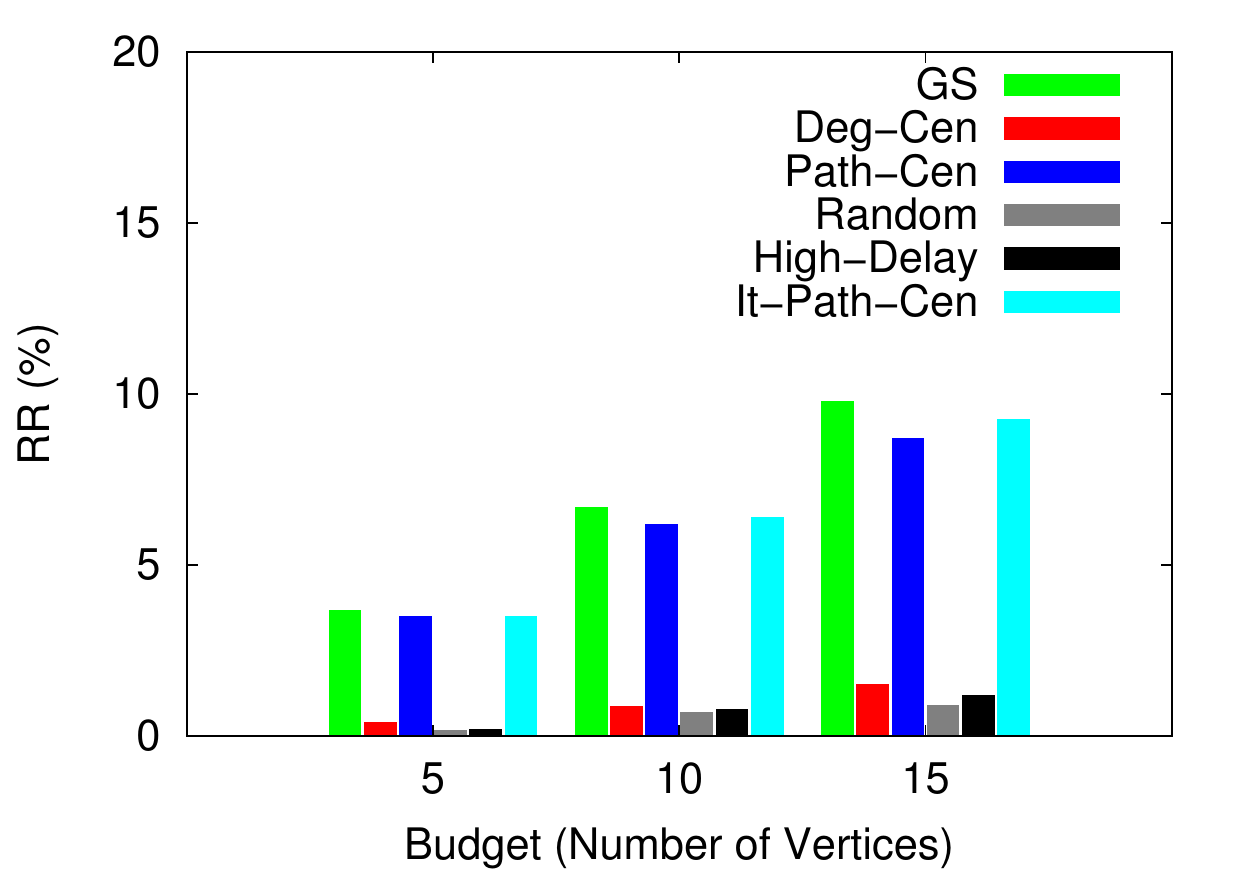}\label{fig:GS_qual_traffic_sp74}}
 \vspace{-2mm}
    \caption{\small Comparison of baselines on Traffic: (a) PCS in the \textit{Uniform Model}; and (b) GS in the \textit{General Model}.  \label{fig:baselines_small_large}}
\vspace{-4mm} \end{figure}

Fig.~\ref{fig:PCS_qual_traffic} presents the RR of competing techniques on the Traffic network with uniform delays using $50log(n)$ samples for PCS. The baseline algorithm It-Path-Cen for the setting of uniform delays is equivalent to the exhaustive greedy GR and this comparison is already available in Fig.~\ref{fig:GSGR_comparison}. On this relatively small network, PCS produces at least $6\%$ better RR than the best alternative Path-Cen. Note, that in this setting simple alternatives such as Random and Deg-Cen, although fast, have unacceptably low quality.

Next we associate the delays (general model) at road intersections (nodes) measured at different times,  and compare with competing techniques. As the results on different snapshots are similar, we show a representative figure on quality (fig. \ref{fig:GS_qual_traffic_sp74}). Using $10log(n)$ samples, GS produces higher RR than both Path-Cen and It-Path-Cen, with up to 1 and 2 orders of magnitude running time improvement respectively (plots omitted due to space constraint). 
Unlike It-path-Cen, GS does not target nodes only based on the number of shortest paths through them, but estimates the improvement of nodes given those already in the target set and achieves a better quality.

In larger graphs computing the exact quality (reduction of SPD) has high computational cost as it requires computing all-pair shortest paths. Hence, in order to evaluate the competing techniques, we estimate RR based on a representative sample of pairwise shortest path lengths. We randomly sample $1000$ pairs $10$ times and average the quality results. We evaluate the competing techniques on DBLP and the Twitter datasets.

First we evaluate the running time in comparison to the best-quality competing techniques in Tab.~\ref{tab:scale_dilkina}. As expected based on their theoretical complexity, Path-Cen and It-Path-Cen~\cite{dilkina2011} do not scale well for large datasets. Our algorithms complete in at most $36$ min, while the alternatives take close to or more than $5$h on the same input (DNF stands for ``does not finish in $5$ hours''). Twitter-50K in this experiment is a subgraph of the Twitter-Politics network involving 50K nodes, while in the uniform-delay setting we evaluate PCS on a subgraph of DBLP of 100K nodes (DBLP-100K). 

Since the methods by Dilkina et al.~\cite{dilkina2011} do not scale for large graphs we compare the quality of our sampling schemes with that of Deg-Cen and High-Delay on the full large-graph datasets (High-Delay is replaced by Random in the uniform model experiments as delays in this setting are equal). To enable even higher scalability for GS, we use multi-threading with $4$ threads to compute the shortest paths (steps $4-7$ in Alg.~\ref{algo:GS}). 
For the rest of the experiments, we use GS(4T).

\begin{table} [t]
\centering
\scriptsize
\begin{tabular}{| c | c | c | c |c |}
\hline
\textbf{Data}& PCS& GS & Path-Cen & It-Path-Cen\\
\hline
DBLP-100K (unif.)&$2$ m& $-$ & $4.5$ h & DNF\\
\hline
Twitter-50K (gen.)&$-$  &$36$ m& DNF & DNF\\
\hline
\end{tabular}
\vspace{-1mm}
\caption{\small Running time comparison of our algorithms and those proposed in~\cite{dilkina2011} (budget = $5$).}\label{tab:scale_dilkina}
\vspace{-3mm}
 \end{table}

\begin{table} [t]
\vspace{-1mm}
\centering
\scriptsize
\begin{tabular}{| c | c | c | c |}
\hline
\textbf{Data}& Algo. & $\#$Sample & Time(min)\\
\hline
Twitter-Celeb (unif.)& PCS & $148$ & $2.5$\\
\hline
Twitter-Politics (unif.)& PCS & $166$ & $3.5$\\
\hline
DBLP (unif.)& PCS & $200$ & $9$\\
\hline
Twitter-Celeb (gen.)& GS(4T) & $148$ & $1$\\
\hline
Twitter-Politics (gen.)& GS(4T) & $64$ & $19$\\
\hline
Twitter-Science (gen.)& GS(4T) & $64$ & $12$\\
\hline
DBLP (gen.)& GS(4T) & $40$ & $62$\\
\hline
\end{tabular}
\vspace{-1mm}
\caption{\small Running times of PCS and GS(4T) with budget = $5$.}\label{tab:time_GS(4T)_PCS}
\vspace{-2mm}
\vsc \end{table}

Tab.~\ref{tab:time_GS(4T)_PCS} presents the running times of our algorithms in both the uniform and general delay settings together with the number of sampled pairs of each run over the full networks. The number of samples $clog(n)$ depends on both the size of the network and the constant $c$ (which we set to values not exceeding $20$). In the uniform scenario (datasets denoted \textit{unif.}), we assume delay $1$ associated with nodes. PCS completes in the order of minutes in uniform-delay networks and GS within 62 minutes on the largest DBLP dataset.

Figs.~\ref{fig:twitter_celeb_qual}-\ref{fig:dblp_random_qual} show the quality of GS in Twitter and DBLP. In all cases GS performs better than alternatives for increasing budget, since the alternatives fail to capture the dependency between upgraded nodes and are limited to local node properties. We get higher quality in Twitter-Celeb as we use relatively higher number of samples. The RR in DBLP is relatively low due to the large network size and disproportionately small budgets ($5$ and $10$ out of $1.1M$ nodes). Fig.~\ref{fig:PCS_budget5} presents an analogous comparison for uniform delay. Our technique PCS outperforms alternative in Twitter (budget $k=5$). In DBLP, Deg-Cen has similar quality to that of PCS since authors of high degree tend to be central.

The only parameter in our techniques is the number of samples which provides a natural trade-off between running time and quality. Our analysis shows that we usually need only small fraction of sampled pairs to match the performance in greedy in both real-world and synthetic data. Details of this analysis are available in the Appendix.





\vspace{-2mm}
\section{Previous Work}

Paik et al.~\cite{paik1995} first introduced a set of design problems in which vertex upgrades improve the delays of adjacent edges. 
Later, Krumke et al.~\cite{krumke1998} generalized this model assuming varying costs for vertex/edge upgrades and proposed to minimize the cost of the minimum spanning tree. Lin et al.~\cite{lin2015} also proposed a delay minimization problem with weights associated with undirected edges. The above formulations are different from ours as in our case delays are associated with vertices. The problems considered in Dilkina et al.~\cite{dilkina2011} are closer to our setting, in that they correspond to a general version of DMP. As we show in our comparative evaluation, our methods dominate those proposed by the authors in both scalability and quality (Section \ref{sec:baselines}). 

Delay minimization and other global objectives (vertex eccentricity, diameter, all-pairs shortest paths etc.) have been previously addressed by \emph{edge addition} \cite{meyerson2009,papagelis2011,parotisidis2015selecting,demaine2010,perumal2013}. 
Meyerson et al.~\cite{meyerson2009} designed approximation algorithms for single source and all pair SP minimization. Demaine et al.~\cite{demaine2010} minimize a network diameter and node eccentricity by adding shortcut edges with a constant factor approximation algorithm. 
Prior work also considers eccentricity minimization in a composite network where a social node connectivity is improved by additional communication network edges~\cite{perumal2013}. All the above problems, however, are based on adding new edges i.e., structural modification, and hence are complementary to our setting. In different applications, node-based and edge-based schemes could be adopted individually or in unison.  

Other related problems involve efficient computation of betweenness centrality. In~\cite{riondato2014}, the authors compute top $k$ nodes based on betweenness centrality via sampling. The group betweenness problem has been solved in almost linear time by Yoshida et al.~\cite{yoshida2014} by a high quality probabilistic approximation algorithm. We develop  similar sampling schemes for a different problem of network design where the metric is based on shortest path centrality of vertices. 
\vspace{-1mm}
\section{Conclusions}
\vspace{-1.5mm}
In this paper, we studied and proposed solutions for the network design problem of node delay minimization. The problem has diverse applications in a variety of domains including social, collaboration, transportation and communication networks. We proved that the problem is NP-hard even for equal node delays. We proved approximation guarantees for a restricted formulation via randomized schemes based on VC dimension theory. We proposed and evaluated high-quality methods for the problem based on sampling that scale to large million-node instances and consistently outperform existing alternatives. We evaluated our approaches on several real-world graphs from different genres. 
We achieved up to $2$ orders of magnitude speed-up compared to alternatives from the literature on moderate-size networks, and obtained high-quality results in minutes on large datasets while competitors from the literature require more than $4$ hours.


\section{Acknowledgments}
Research was sponsored by the Army Research Laboratory and accomplished under Cooperative Agreement Number W911NF-09-2-0053 (the ARL Network Science  CTA). The views and conclusions in this  document are those of the authors and should not be interpreted as representing the official policies, either expressed or implied, of the Army Research Laboratory or the U.S. Government. The U.S. Government is authorized to reproduce and distribute reprints for Government purposes
notwithstanding any copyright notation here on. We also would like to thank Arlei Silva for helpful discussions.

\bibliographystyle{abbrv}
\bibliography{icdm_extended}  

\section{appendix}

\textbf{Proof of Theorem 1}
\begin{proof}
We outline a reduction from the Set Cover problem. Consider an instance of the NP-complete Set Cover problem, defined by a collection of subsets $S_{1},S_{2},...,S_{m}$ for a universal set of items $U=\{ u_{1},u_{2},...,u_{n} \}$. The problem is to decide whether there exist $k$ subsets whose union is $U$. 
To define a corresponding DMP instance, we construct an undirected graph with $n+m+mp$ nodes: there are nodes $i$ and $j$ corresponding to each set $S_{i}$ and each element $u_{j}$ respectively, and an undirected edge $(i,j)$ whenever $u_{j}\in S_{i}$. Every $S_{i}$ is connected to $S_{j}$ when $i \neq j$ and  $i,j \in {1,2,...,m}$. Every $u_{i}$ is connected to other $u_{j}$ when $i \neq j$ and  $i,j \in {1,2,...,m}$. There are $p$ vertices (with degree $1$) attached to each $S_{i}$. The $j$-th vertex (among these $p$ vertices) attached with every $S_i$ makes the set $A_{j}$. All vertices have delay of $1$. Intuitively, the construction makes the vertices in set $S$ more likely to be chosen in TS. 

It is easy to see that vertices in $A_i$ will not be in TS. Next we prove that the minimum reduction (quantity $A$) by any vertex from $S$ is larger than the maximum reduction (quantity $B$) by vertices from $U$. As the delays are $1$, the reduction depends on number of shortest paths. The maximum  number of shortest paths (quantity $A$) that pass through any vertex in $U$ after being chosen in TS is less than the minimum (quantity $B$) of the same through any vertex in $S$. Quantity $A$ is exactly $(n-1)(m+pm)$, while $B$ is $p(m+n+p(m-1))$. A choice of $p=mn$ makes $B$ larger than $A$. Hence a choice of nodes from $S$ is always preferable.
 
 Tabs. \ref{tab:tableforcoveruniform} and \ref{tab:tablenoncoveruniform} show \textit{SPD} computation between nodes in the sets $S,U,$ and $A_i$ in two different cases. The quantities are as follows: \small $W_1=k(2(m-k)+k-1)+(m-k)(3(m-k-1)+2k),W_2=k(2(m-k)+k-1+1)+(m-k)(3(m-k-1)+2k+2),  W_3=kn+2(m-k)n, W_4=k(m-k)+(m-k)(1(k+1)+2(m-1-k)), W_5= km+2(m-k)(m-1)+1(m-k)$.
 
 \normalsize
 Our claim is as follows: the Set Cover problem is equivalent to deciding if there is a set of $k$ vertices whose upgrade leads to $SPD\leq X$ (where $X=(m+n-1) (-k+m+n)+m p^2 (3 m-2 k)+p (k (4-3 m)+m (4 m-5))+2p(2mn-kn)$, sum of all the elements in Tab. \ref{tab:tableforcoveruniform}). 
For a ``yes''-instance of the Set Cover problem, we show that all $k$ TS vertices correspond to selected sets in the set cover and we achieve the effective SPD of $X$. 
For a ``no''-instance of the corresponding Set Cover problem,
the argument is as follows. For the TS we choose $k$ nodes from set $S$ as we have already proved, no vertex from set $U$ or $A_i$ can be in TS. For a ``no''-instance Table \ref{tab:tablenoncoveruniform} shows the desired \textit{SPD}. Comapring two tables, it is evident that the \textit{SPD} in Table \ref{tab:tablenoncoveruniform} is greater than $X$. If the corresponding Set Cover problem has a set cover $\leq k$, then only the \textit{SPD} is reduced to $X$. Hence the claim is true and the problem is NP-hard.
 In conclusion, DMP is NP-hard under the general model.
 \end{proof}
  \begin{table}[ht]
\vspace{-2mm}
\centering
\scriptsize
\begin{tabular}{ |c|c|c|c|c|c|c|}
\hline
$ *** $ & $A_1$ & $A_2$&$...$&$A_p$& $S$ & $U$  \\ 
\hline
 $A_1$ & $W_1$ & $W_2$&$...$&$W_2$& $W_5$ & $W_3$ \\ 
\hline
$A_2$ & $W_2$ & $W_1$&$...$&$W_2$& $W_5$ & $W_3$ \\ 
\hline
$...$ & $...$ & $...$&$...$&$...$& $...$ & $...$ \\ 
\hline
$A_p$ & $W_2$ & $W_2$&$...$&$W_1$& $W_5$ & $W_3$ \\ 
\hline
$S$ & $W_4$ & $W_4$&$...$&$W_4$& $(m-1)(m-k)$ & $n(m-k)$ \\ 
\hline
$U$ & $W_3$ & $W_3$&$...$&$W_3$& $mn$ & $(n-1)n$ \\ 
\hline
\end{tabular}
\vspace{-2mm}
\caption{ Sum of shortest path delays when size of set cover is  $\leq k$ and delays of nodes which forms the set cover are reduced to $0$. \label{tab:tableforcoveruniform}}
\vsb \end{table}

\begin{table}[ht]
\centering
\scriptsize
\begin{tabular}{ |c|c|c|c|c|c|}
\hline
$ *** $ & $A_1$ & $A_2$&$A_p$& $S$ & $U$  \\ 
\hline
$A_1$ & $W_1$ & $W_2$&$W_2$& $W_5$ &  $\geq W_3$ \\ 
\hline
$A_2$ & $W_2$ & $W_1$&$W_2$& $W_5$ &  $\geq W_3$ \\ 
\hline
$...$ & $...$ & $...$&$...$& $...$ & $...$ \\ 
\hline
$A_p$ & $W_2$ & $W_2$&$W_1$& $W_5$ &  $\geq W_3$ \\ 
\hline
$S$ & $W_4$ & $W_4$&$W_4$& $(m-1)(m-k)$ & $>n(m-k)$ \\ 
\hline
$U$ & $\geq W_3$ & $\geq W_3$&$\geq W_3$&  $>mn$ & $(n-1)n$ \\ 
\hline
\end{tabular}
\vspace{-2mm}
\caption{Sum of shortest path delays when size of set cover is $>k$ and delays of arbitrary k nodes from $S$ set to $0$. 
\label{tab:tablenoncoveruniform}}
\vsa
\vsb \end{table}

\vspace{3mm}
\textbf{Approximability: Probabilistic approximation}

Let $(U,R)$ be a set system, where $U$ is a finite set and $R$ is a collection of subsets of $U$. A set $W$ is \textit{shatterable} in $R$ if and only if for any subset $W'$ of $W$, there exists $R_i\in R$ such that $W\cap R_i=W'$. The VC-dimension of the set system is defined as the largest integer $d$ such that no subset of $U$ of size $d+1$ can be shattered. In addition, given a parameter $\epsilon \in [0,1]$, a set $U' \subset U$ is called an \textit{$\epsilon$-net} on $(U,R)$ if for any set $R_i\in R, |R_i|\geq \epsilon|U|$, $U'\cap R_i \neq \emptyset$. Intuitively, an \textit{$\epsilon$-net} is a set such that any sufficiently large subset (parametrized by $\epsilon$) in the set system has common elements with. Next, we introduce some established results from VC theory that we use in our analysis.
\begin{lemma}
\label{lemma:net}
 \textbf{$\epsilon$-net~\cite{haussler1986}:} For any set system with bounded VC-dimension $d$, a randomly drawn sample of size $O(\frac{d}{\epsilon}log\frac{d}{\epsilon}+\frac{1}{\epsilon}log\frac{1}{\delta})$ is an $\epsilon$-net with probability $\delta$.
\end{lemma}

We analyze the VC dimension of shortest paths in a graph. Given a graph $G=(V,E)$, a shortest path can be uniquely defined by a source vertex $s$ and a destination vertex $t$. Let $R$ be a set system of unique shortest paths. This system is defined as follows: if any vertex pair $(u, v)$ is contained in two shortest paths $p_{s_1,t_1},p_{s_2,t_2}\in R$, then $u$ and $v$ are linked by the same path in $R$. We introduce an important lemma:

\begin{lemma}
\label{lemma:dimension_usps}
\textbf{USPS \cite{ruan2011}:} There exists a unique shortest path system (USPS) for every graph. \textbf{Dimension \cite{abraham2011,tao2011}:} For a graph $G = (V,E)$, the set system $(V,R)$, where $R$ is a unique shortest path system, has a VC-dimension of $2$. 
\end{lemma}



We next show an approximation for ``long'' paths of delay $\epsilon n$ or higher, where $|V|=n$ and $0 < \epsilon < 1$. 
Our optimization objective for ``long'' paths is $SPD^\epsilon(G)=\Sigma_{s,t\in V, d(s,t)\geq \epsilon n}d(s,t)$. 
Let $R^\epsilon_{opt}(k)$ represent the reduction in $SPD^\epsilon$ by the optimal TS of size $k$ and $R^\epsilon_{rand}(kb)$, the reduction due to $kb$ \emph{randomly chosen} vertices. The relationship between $k$, $b$ and $\epsilon$ is captured in the following theorem.

\begin{thm}
\label{thm:vc}
Given a confidence parameter $\delta$,
$ \frac{R^\epsilon_{opt}(k)}{R^\epsilon_{rand}(kb)}\leq k$ with probability $\delta$, where $kb=(\frac{2}{\epsilon}log\frac{2}{\epsilon}+\frac{1}{\epsilon}log\frac{1}{\delta})$.
\end{thm}

\begin{proof}
We develop an upper bound for $R^\epsilon_{opt}(k)$. In the best case scenario, all $k$ nodes of the TS are present in a set $S$ of all different shortest paths. So, the maximum reduction is $k|S|$ (as every vertex has a delay of $1$). From Lemma $\ref{lemma:dimension_usps}$, the USPS of a graph always exists and its VC-dimension is $2$. A random sample of $b'(\frac{2}{\epsilon}log\frac{2}{\epsilon}+\frac{1}{\epsilon}log\frac{1}{\delta})$ vertices will be an $\epsilon$-net with probability $\delta$ by Lemma $\ref{lemma:net}$ ($b'$ is the constant in the asymptotic bound in Lemma $\ref{lemma:net}$ and $1/b'= b$). In other words, at least one vertex from every shortest path (of length $\geq \epsilon n$) of the USPS will be included in the sample. If these vertices are selected as a target set, the resulting reduction will be $|S^*|$, where $S^*$ is the set of shortest paths of length $\geq \epsilon n$ in USPS. It is clear that $|S^*|\geq|S|$. So, $\frac{R^\epsilon_{opt}(k)}{R^\epsilon_{rand}(kb)}=\frac{k|S|}{|S^*|}\leq k$.
\end{proof}

The theorem shows that the problem of minimizing the delay of long paths under the uniform model has an approximation of $k$. As the bound on path length $\epsilon n$ increases, we need a smaller number of samples to ``cover'' every path. Through experiments, we compare our proposed algorithm PCS (Algorithm \ref{algo:PCS}) against the theoretical upper bound on the restricted metric involving only long paths.

  
  
 \vspace{3mm} 
\textbf{Proof of Theorem \ref{thm:pathcount}}
\begin{proof}
There are three different cases to consider based on the kind of shortest paths. First, for shortest paths where $v$ is start vertex, updating its delay results in a reduction of $n-1$. Second, for shortest paths that go through $v$ (but $v$ is not the start vertex), the reduction is $\zeta^v$. Finally, for shortest paths where $s$ is not on the path, the change in its delay does not result in any reduction (due to equal delays). Therefore, $RS(v|S)=\zeta^v+(n-1)$.
\end{proof}

\vspace{3mm}
\textbf{Worst Case of Algorithm 1 (GR) }
\newline
Next, we shortly discuss the worst case behavior of GR. GR can produce arbitrarily bad result under the general model. 
One can show that $\frac{|SP_{GR}-SP_{opt}|}{1+SP_{opt}}=O(n^2)$, where $n$ is the number of vertices in the network and $SP_{opt}$ and $SP_{GR}$ are the corresponding SPD of an optimal algorithm and GR. We show this in a small example and budget of $2$ in Fig.~\ref{fig:greedycounterrandom}. The boxes in the figure represent cliques of $n'$ nodes, where $2n'$ is close to $n$, $2n'>>5$. We ignore the shortest paths among vertices oustside these cliques. While the optimal target set is comprised of the green vertices, GR might choose the red nodes and achieve the described result. We have also found evidence of a non optimal performance of Greedy in case of the uniform model. That evidence tells us that GR can produce a constant factor of the optimal result for a low budget.


\begin{figure} [t!]
\centering
\begin{tabular}{c}
\includegraphics[ height=1.5in,width=2.8in,angle=0]{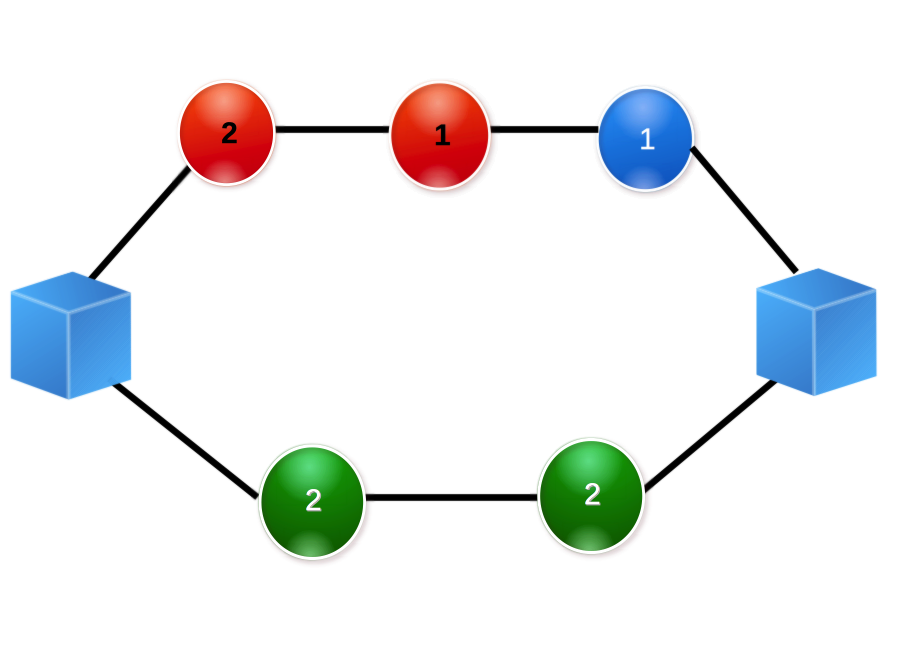}

\end{tabular}
\caption{ Non optimal performance of Greedy (GR)}
\label{fig:greedycounterrandom}

\end{figure}


\vspace{3mm}
\textbf{Algorithm 3 (PCS)} \\
Algorithm \ref{algo:PCS} (PCS) computes TS based on estimates of number of shortest paths through each vertex. The approximation error $\epsilon$ bounds the difference between reduction by PCS and Greedy (GR) in each iteration. In each of the $k$ iterations, PCS first samples $p$ pairs of nodes from the population of all pairs (the proven approximation still holds when the samples are obtained before the iteration starts as in Alg. \ref{algo:GS}). It computes the shortest distances between the vertices of each pair in the sample and thus finds an approximate measure of number of shortest paths through each vertex. If both the vertices of a pair have delay $1$, they both are present in edited graph (step 20). The edited graph is obtained by deleting the vertex with delay $0$ and by adding all the edges between its neighbours. Computing BFS explores the vertices on all possible shortest paths. If one of them has delay $0$, we compute BFS from the other vertex to its gateway vertex (For each such vertex $v$ added to TS, we maintain a list of vertices, called $gateway$, where $v.gateway=\{u\ |\ d(v,u) = 0, l(u) = 1\}$). If both of them have delay $0$, we compute BFS between each of the gateway vertices of them. We choose the pair(s) of minimum distance. The shortest paths between them explore the desired vertices (with delay $1$). The vertex with the maximum $\zeta$ is chosen in each iteration. Next we explain that the overall complexity is $O(kp(m+n))$.

\vspace{3mm}
\textbf{Running Time of Algorithm 3 (PCS)} \\ 
For steps $5,9,13$, PCS performs BFS. The probability of picking vertices with delay $0$ is $\frac{x^2}{n^2}$, where $x$ is number of vertices in the current TS. So, the expected time complexity for one of these steps is $p((1-\frac{x^2}{n^2})(m+n)+ g\frac{x^2}{n^2}(m+n))=p(m+n)(1+(g-1)\frac{x^2}{n^2})$ where $g\leq xd_m$ is the maximum size of $gateway$ list. $d_m$ is maximum degree among the degree of the vertices in TS. Step $20$ can take $O(d_m'^2)$, where $d_m'$ is maximum degree of vertex among the vertex and its neighbors. Step $21$ takes $O(x^2d_m^2)$ as it accumulates vertices from $gateway$ of the neighbors. So, the running time for one step when the TS has $x$ vertices is $O(p(m+n)(1+\frac{x^3d_m^2}{n^2})+d_m'^2+x^2d_m^2)$. As $x (x\leq k)$ is small, and the running times of step $20$ and $21$ are not tight, the running time is dominated by $p(m+n)$. So, the overall time complexity is $O(kp(m+n))$.


\begin{figure*}[t]
    \centering
    \subfloat[Synth-2K-Uniform]{\includegraphics[width=0.25\textwidth]{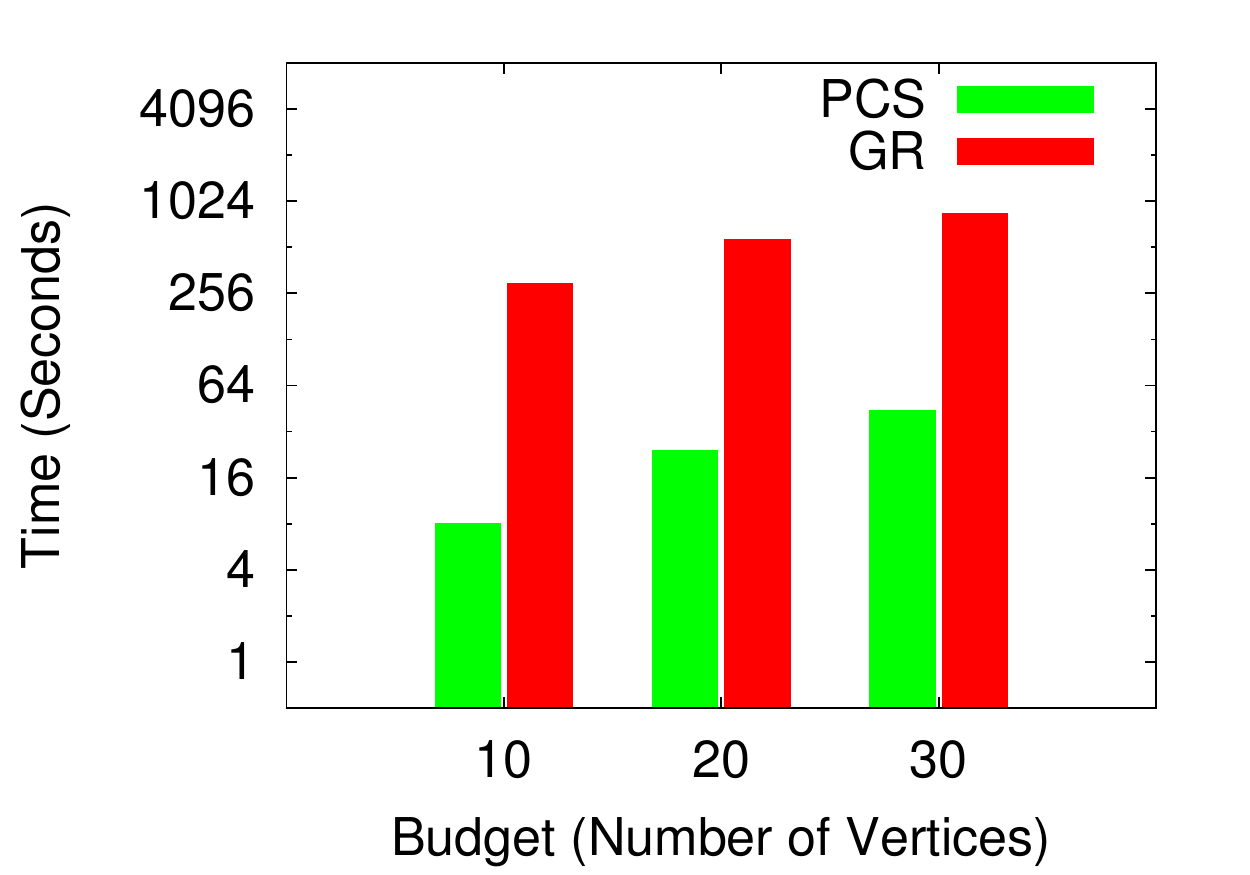}\label{fig:PCSGR_effi_2000}}
    \subfloat[Synth-2K-Uniform]{\includegraphics[width=0.25\textwidth]{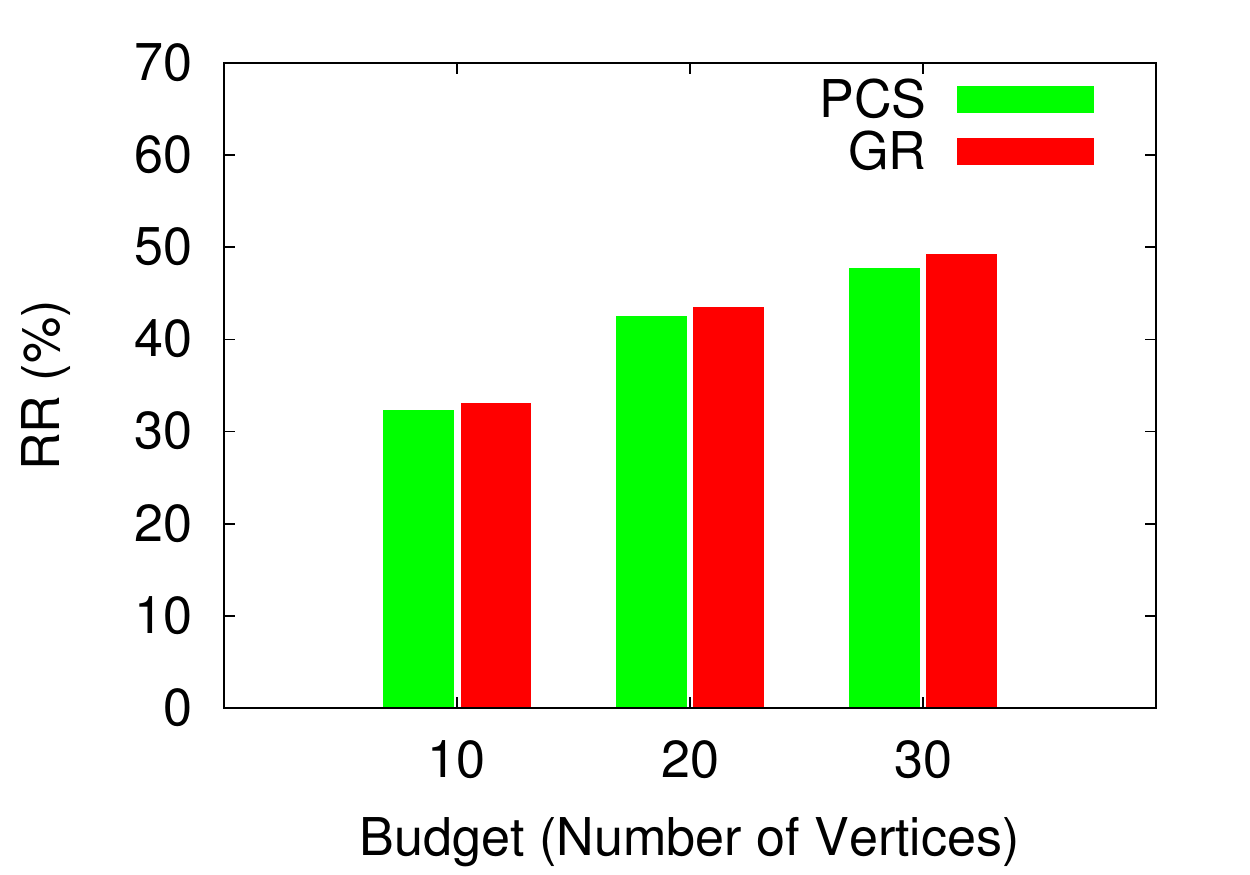}\label{fig:PCSGR_qual_2000}}
    \subfloat[Synth-2K]{\includegraphics[width=0.25\textwidth]{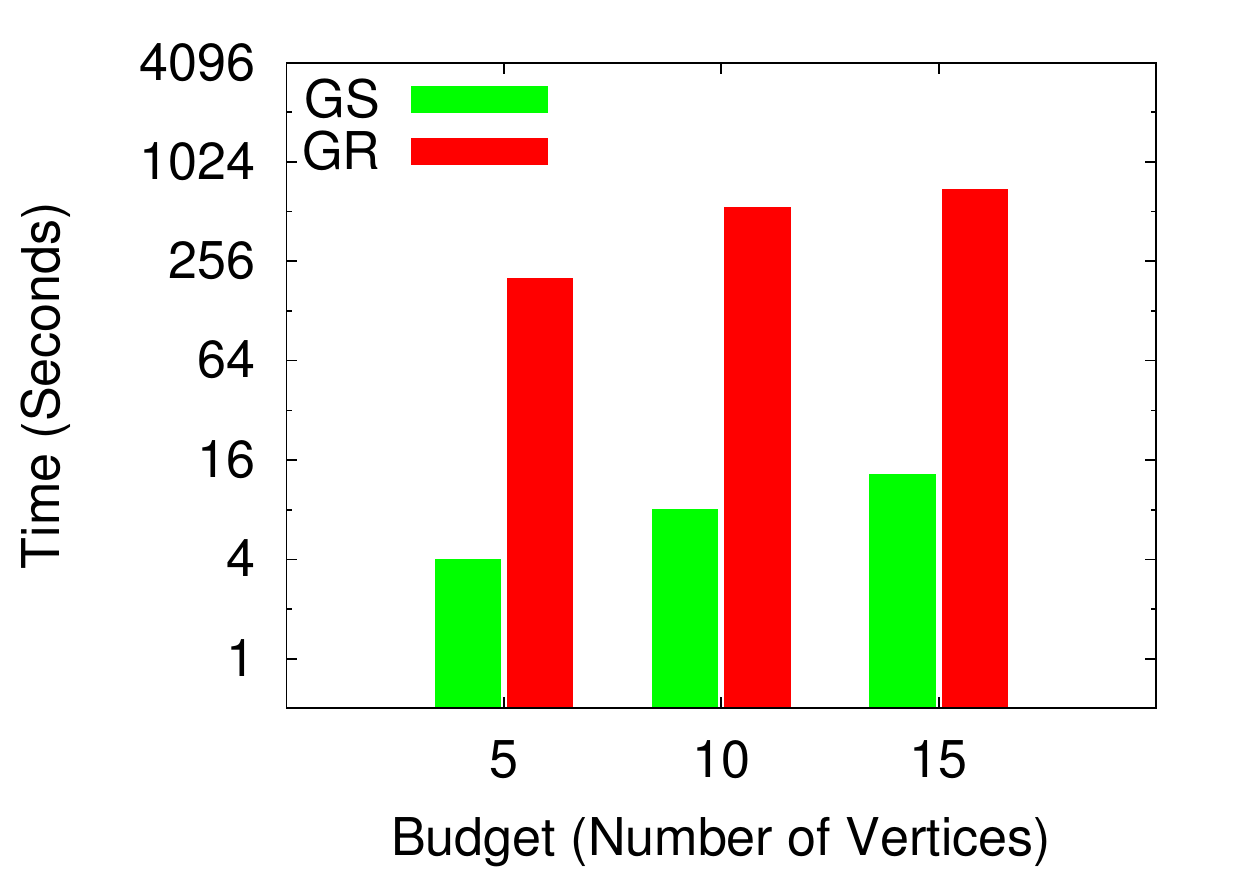}\label{fig:GSGR_effi_2000_500-1000}}
    \subfloat[Synth-2K]{\includegraphics[width=0.25\textwidth]{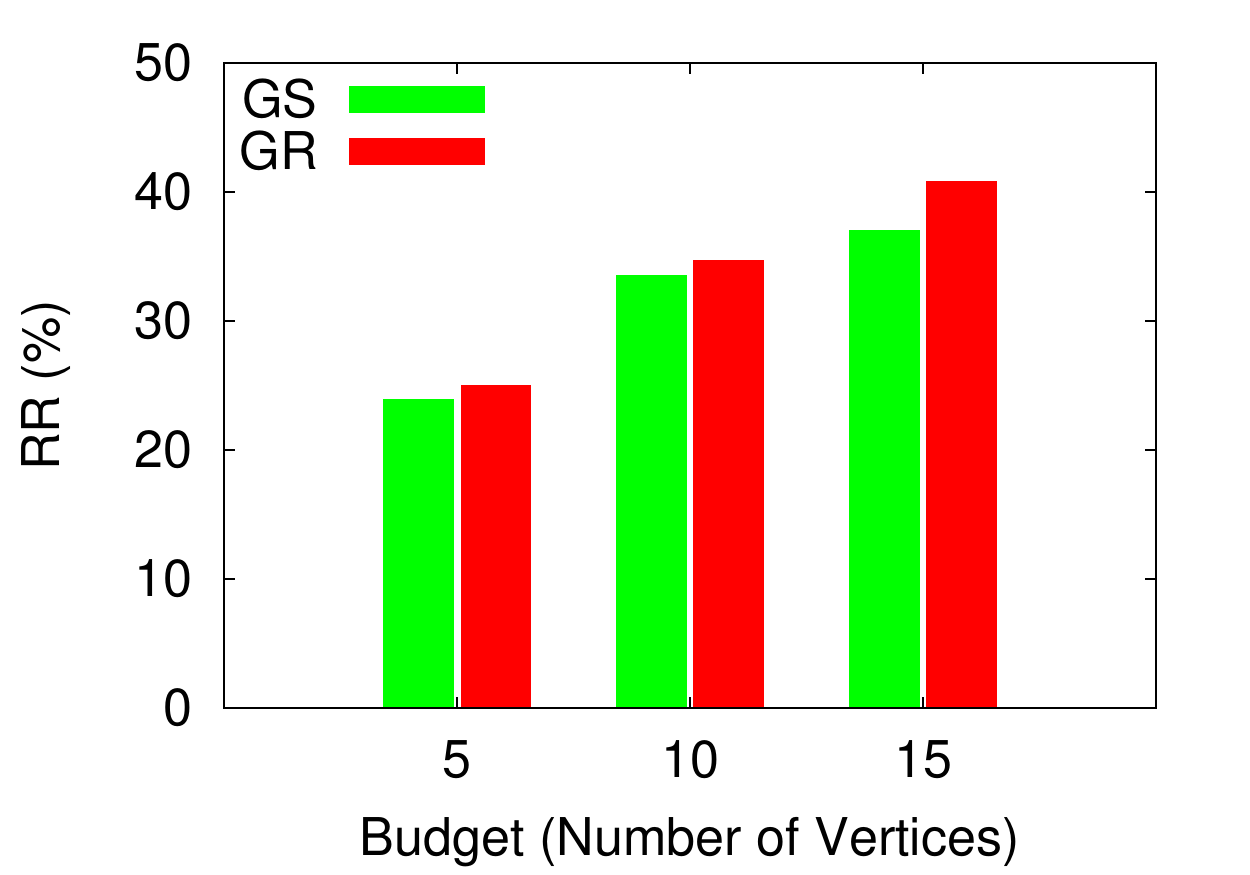}\label{fig:GSGR_qual_2000_500-1000}}
    \caption{\textbf{Uniform model:} (a-b) Execution time and relative reduction comparison between Greedy (GR) and Path Counting (PCS) for synthetic data; \textbf{General model:} (c-d) Greedy (GR) and sampling-based Greedy (GS) for synthetic data (random delays in $[500,1000]$) .\label{fig:GSGR_comparison2}}
\vsb \end{figure*}

\vspace{3mm}
\textbf{Experiments (Synthetic dataset): quality of sampling compared to Greedy}

We also use synthetic small-world networks generated using the Barabasi-Albert (BA) model \cite{barabasi1999}. The assignment of synthetic delays varies and is explained with more details in each experiment.

We report the number of samples in our sampling schemes as $c*logn$, where the sample constant $c$ is related to the expected error $\epsilon$ in Thms.~\ref{thm:approxgeneral} and \ref{thm:approxpathcount}. Unless stated otherwise, we use $c=10$. 

First, we compare our sampling schemes GS and PCS with Greedy (GR) in order to evaluate the effect on quality due to sampling which we theoretically analyze in Thms. \ref{thm:approxgeneral} and \ref{thm:approxpathcount}. To enable the comparison, we use relatively small synthetic and real datasets due to the limited scalability of GR. 
Networks used for evaluation include Traffic and a BA synthetic network with $2000$ vertices and rate of growth: $5$ edges per node. 
The quality of the compared algorithms is quantified as the Relative Reduction (RR) of SPD, while efficiency---in terms of wall-clock time. We use $3.5logn$ samples for GS and PCS in these experiments.

In all experiments, our sampling schemes achieve similar quality as that of Greedy (GR), while taking close to two orders of magnitude less time. In the uniform model, the difference in quality between our sampling scheme PCS and GR does not exceed $1\%$ (Fig.~\ref{fig:PCSGR_qual_2000}), while PCS takes only $2\%$ of the time taken by GR (Fig.~\ref{fig:PCSGR_effi_2000}). 
This trend persists in the case of the general delay model for which we employ our sampling-based Greedy (GS). In synthetic networks with uniformly chosen random delay,
the quality of GS is within $3\%$ of that of GR, while its running time is only $1.5\%$ of the time taken by GR (Fig. \ref{fig:GSGR_effi_2000_500-1000},\ref{fig:GSGR_qual_2000_500-1000}).

\begin{figure*}[t]
\vspace{-3mm}
    \centering
     \subfloat[kb = 100]{\includegraphics[width=0.21\textwidth]{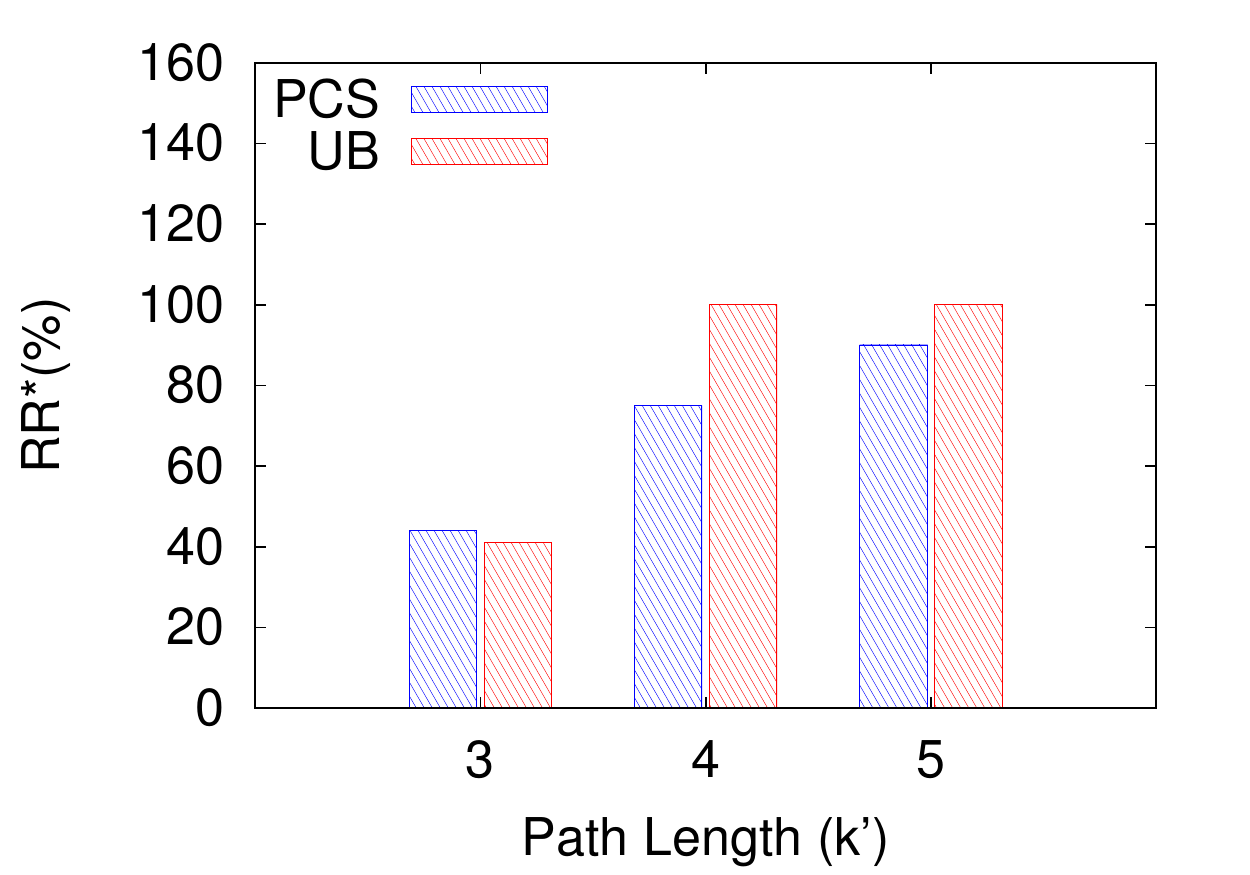}\hspace{-0.2cm}\label{fig:PCS_VC100k}}
    \subfloat[kb = 200]{\includegraphics[width=0.21\textwidth]{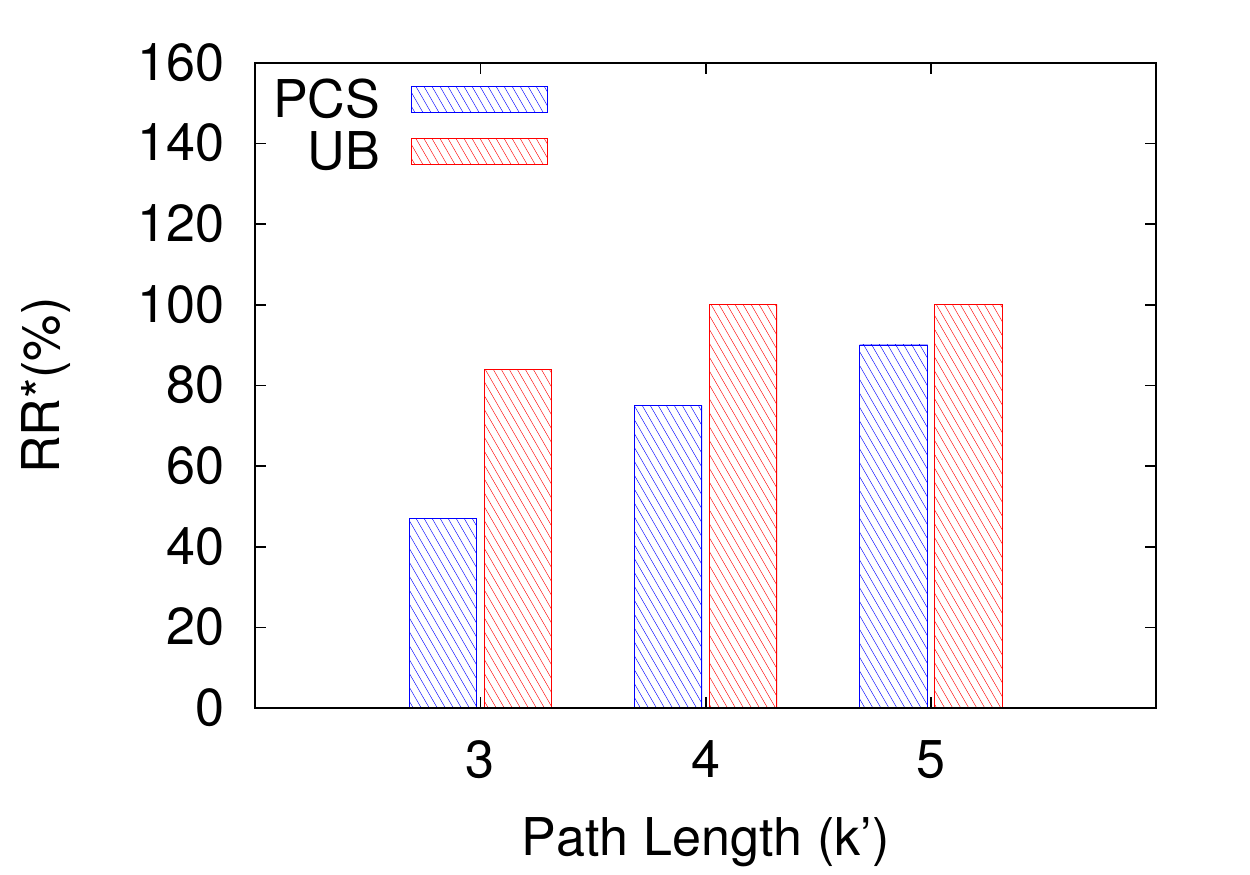}\hspace{-0.2cm}\label{fig:PCS_VC200k}}
    \subfloat[Sum of "all" vs "long"]{\includegraphics[width=0.21\textwidth]{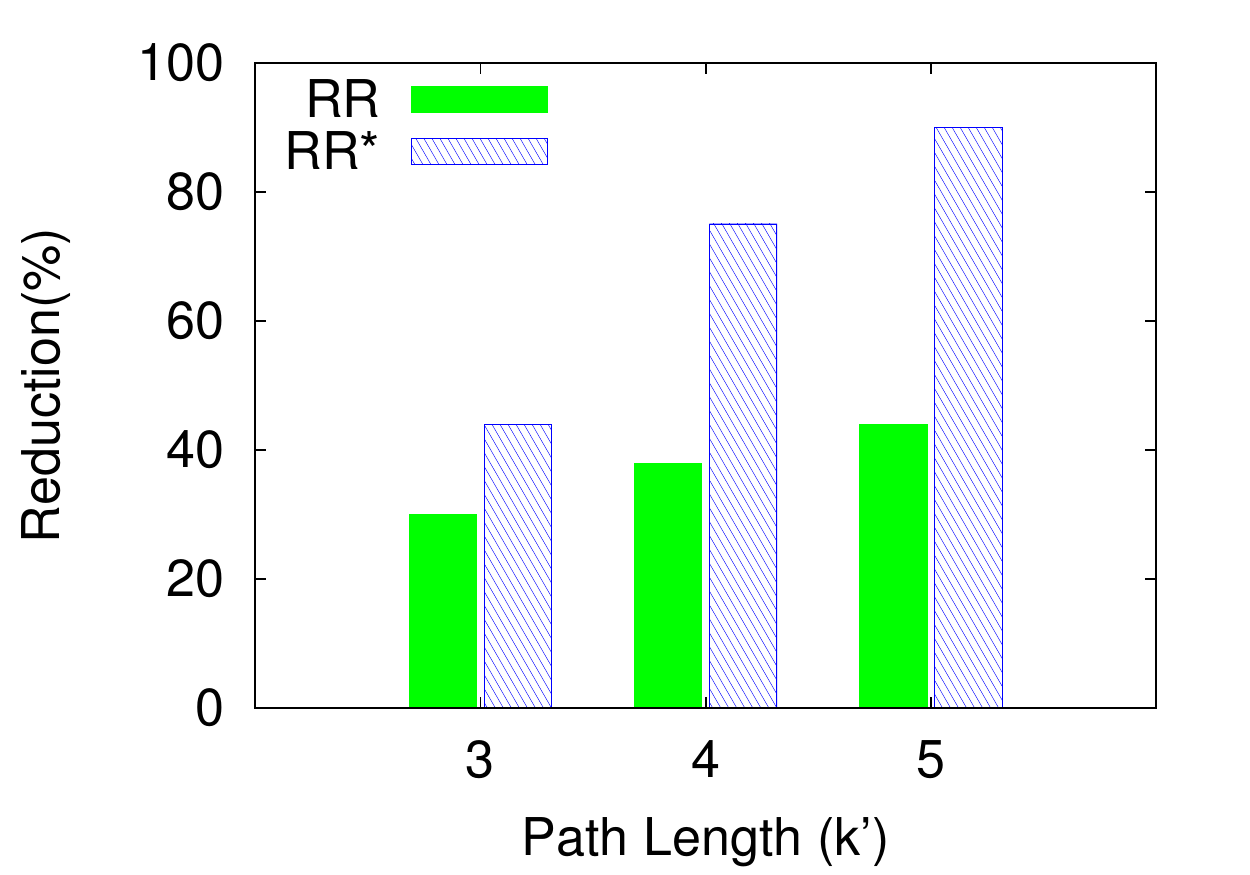}\hspace{-0.2cm}\label{fig:PCS_VC100}}
    \caption{(a-b): Quality of PCS in comparison to the theoretical UB. The horizontal axis represents the budget ($k$) for the restricted metric assuming $k=k'=\epsilon n$. (c): RR and RR* for PCS for the sum of all paths metric and the restricted metric (sum of paths with length $\geq k'$) respectively.
    \label{fig:PCS_others}}
\end{figure*}


\begin{figure*}[t]
    \centering
    \subfloat[Scalability]{\includegraphics[width=0.2\textwidth]{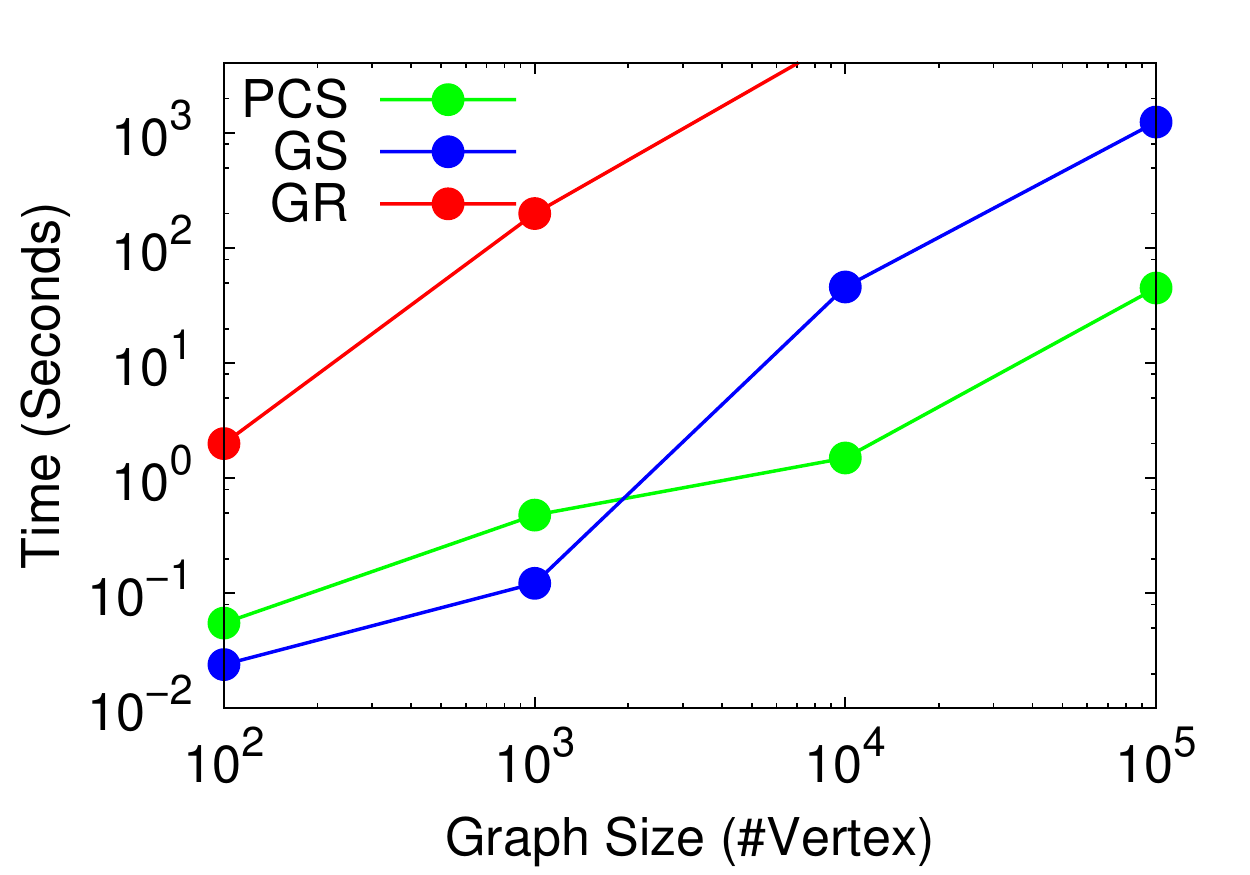}\label{fig:PCS_GS_Scale}}
    \subfloat[PCS: fixed budget]{\includegraphics[width=0.2\textwidth]{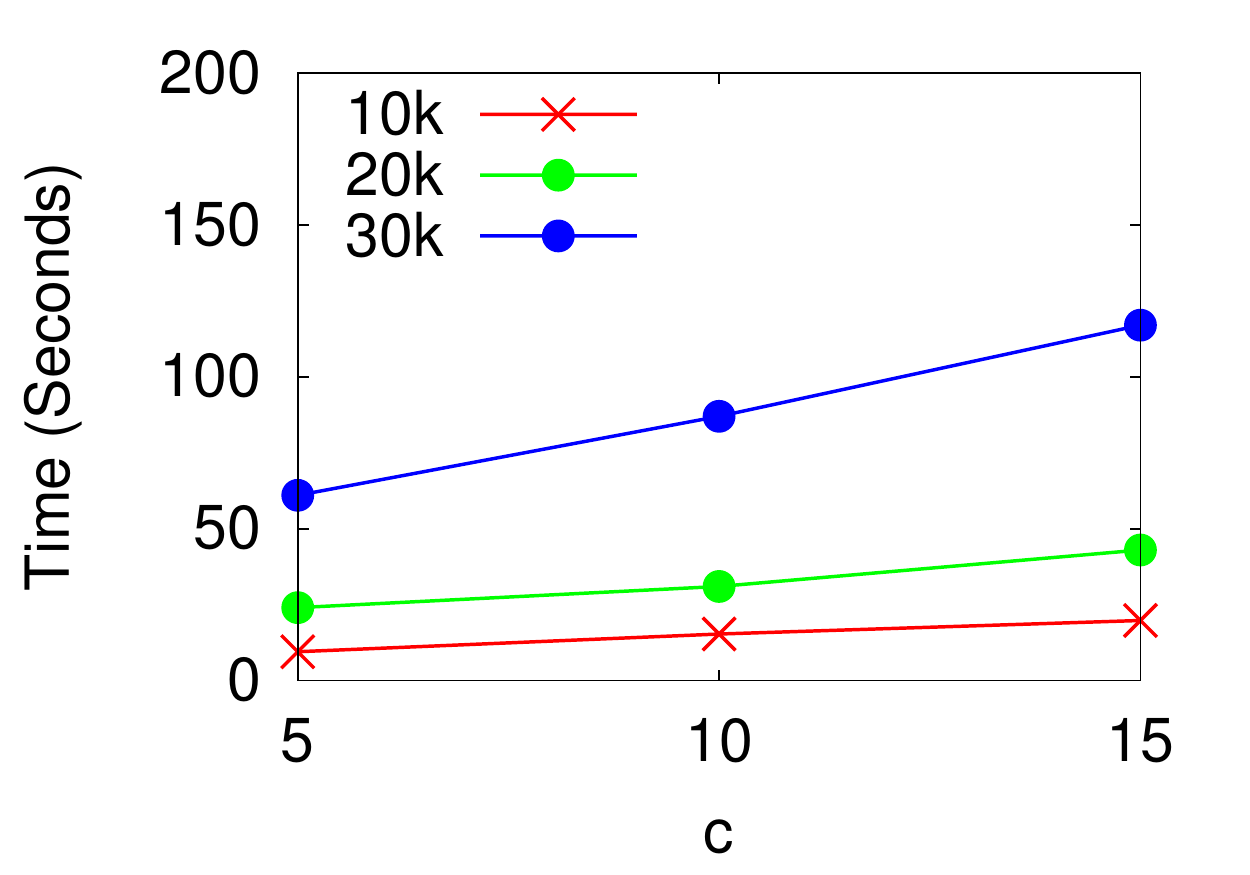}\label{fig:PCS_time_sample}}
    \subfloat[PCS: fixed sample]{\includegraphics[width=0.2\textwidth]{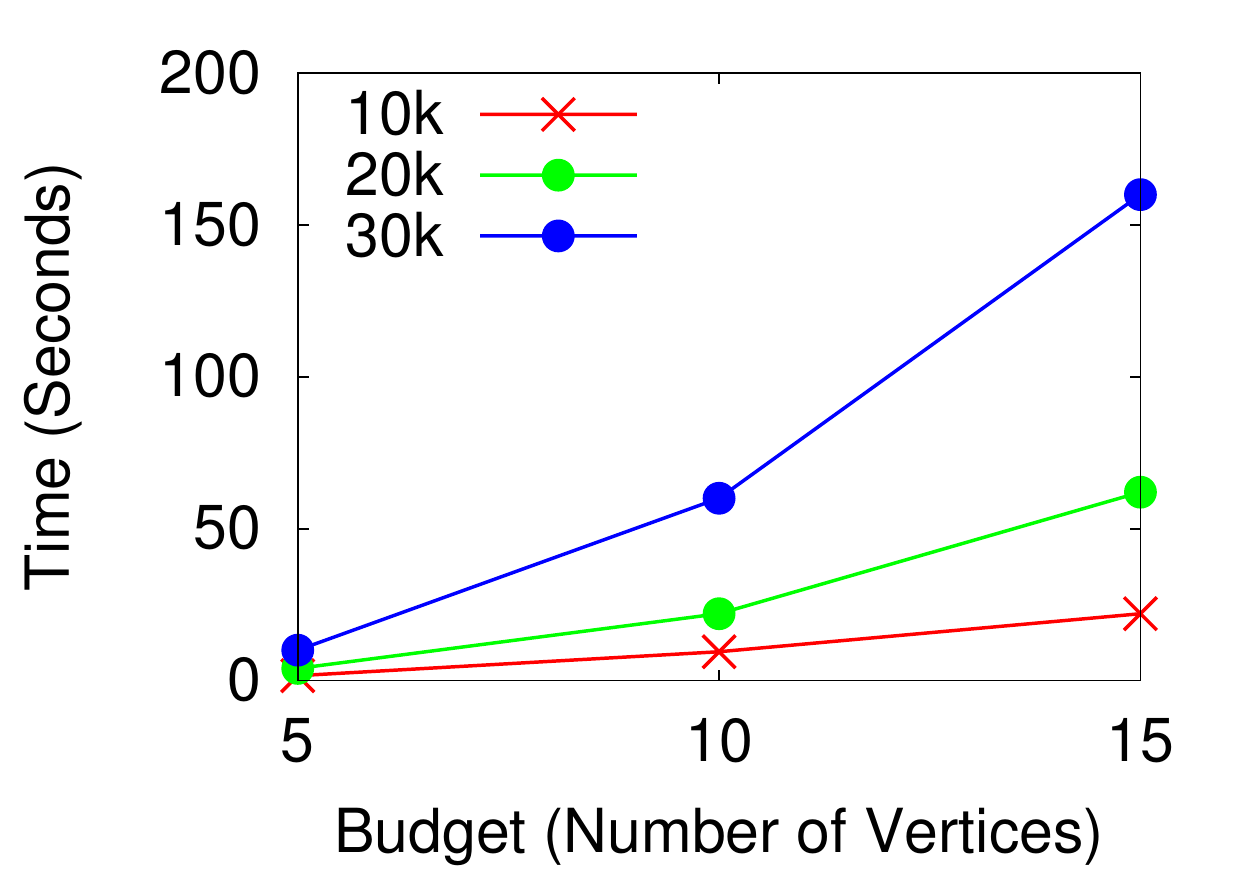}\label{fig:PCS_time_budget}}
    \subfloat[GS: fixed budget]{\includegraphics[width=0.2\textwidth]{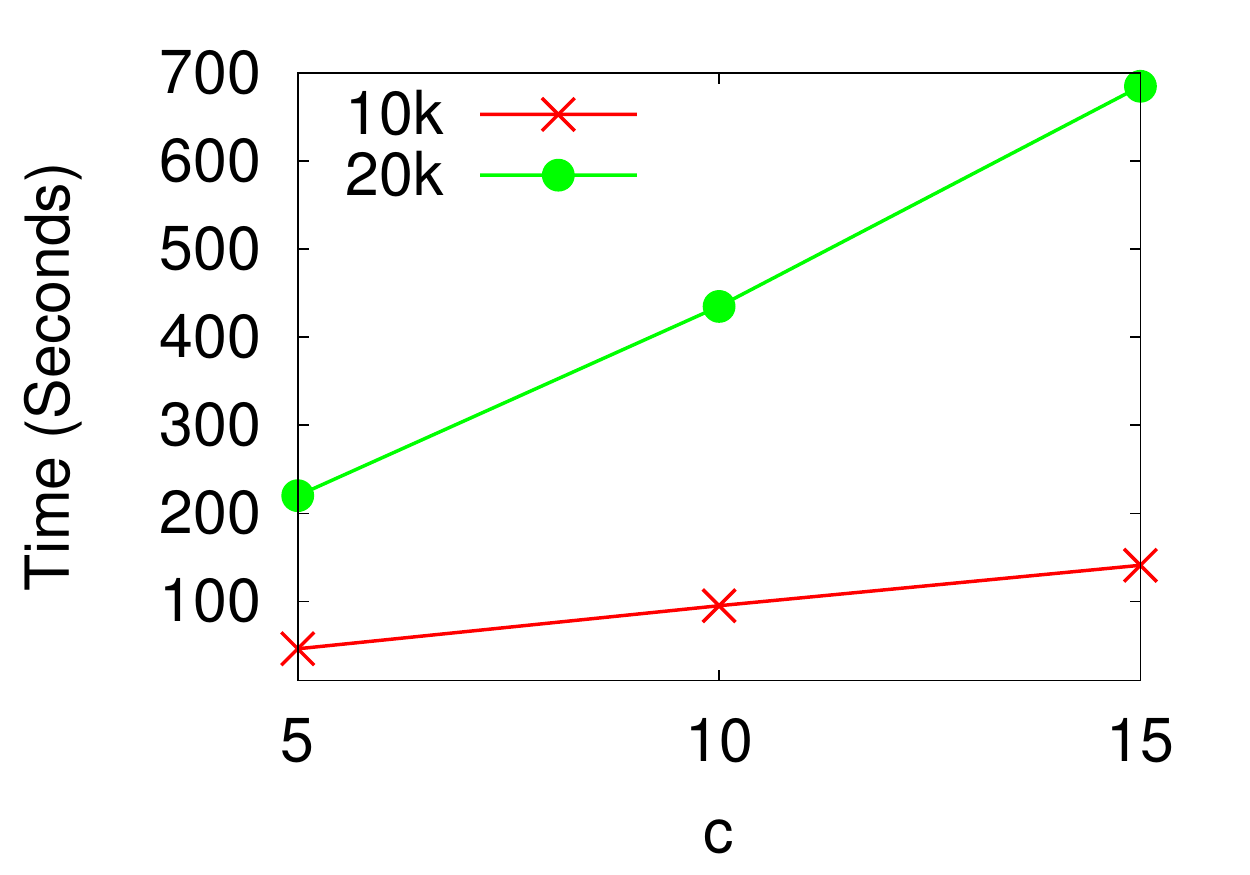}\label{fig:GS_time_sample}}
    \subfloat[GS: fixed sample]{\includegraphics[width=0.2\textwidth]{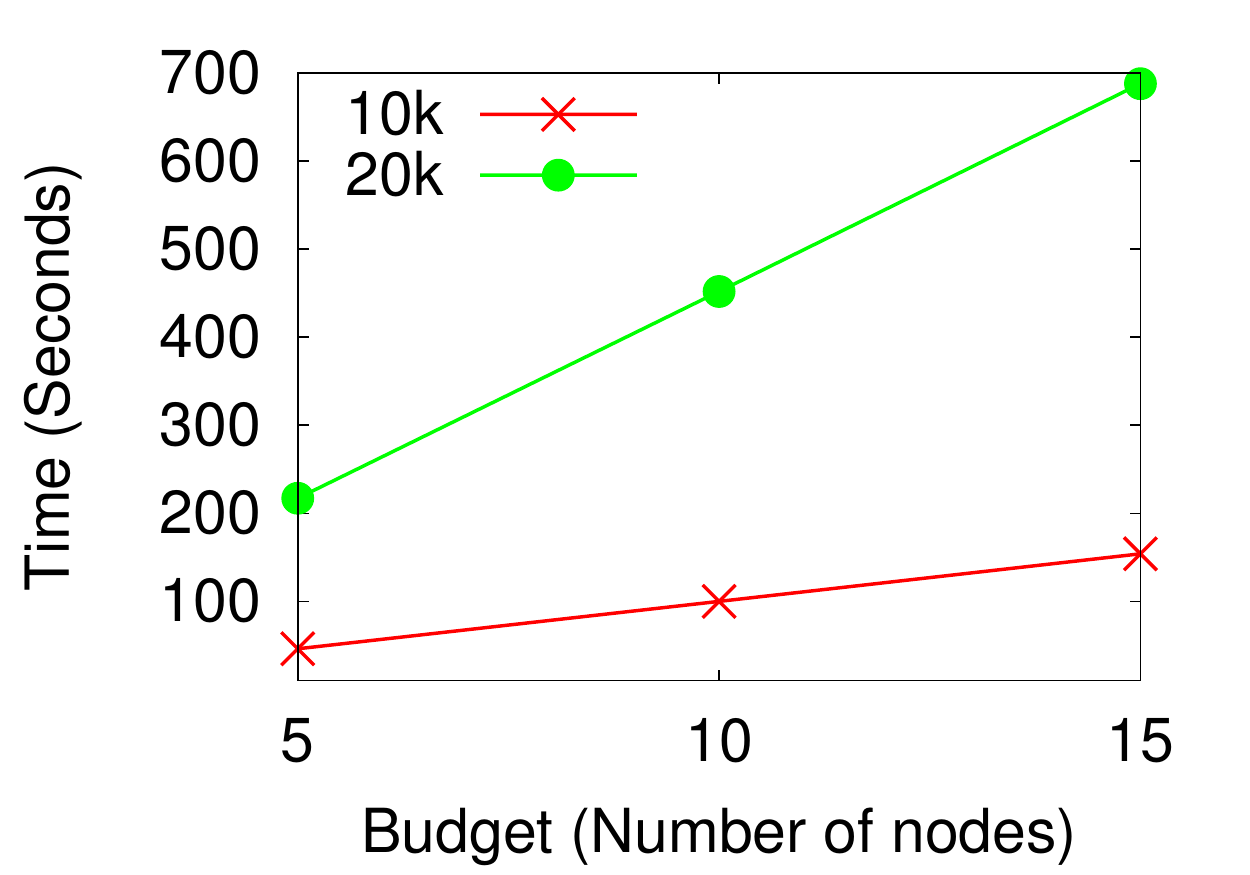}\label{fig:GS_time_budget}}
   \caption{Computation times of PCS and GS on synthetic graphs of increasing size and for varying number of samples and budget. Figures \ref{fig:PCS_time_sample} and \ref{fig:PCS_time_budget} show the behavior for PCS and Figures \ref{fig:GS_time_sample} and \ref{fig:GS_time_budget} show the same for GS. \label{fig:PCS_GS_others}}
\vsc \end{figure*}

\vspace{3mm}

\textbf{Experiments: Comparison of PCS and Upper Bound from Theorem ~\ref{thm:vc} on ``Long'' Paths} 

Theorem~\ref{thm:vc} compares the performance of a random algorithm against the optimal algorithm for the restricted metric of sum of ``long'' paths (length $\geq k' =\epsilon n$). Based on the theorem, $k*R^\epsilon_{rand}(kb)\geq R^\epsilon_{opt}(k)$. We can, thus, compare PCS($k$) (PCS with budget $k$) against $k*R^\epsilon_{rand}(kb)$ as a proxy for comparing against $R^\epsilon_{opt}(k)$. Since the constant $b$ is not known, we vary it to evaluate the quality of PCS($k$). 


We experiment with two different settings for $kb$ ($100, 200$) in a $2,000$-vertex subgraph of the DBLP data. We assume that the path length threshold and the budget are the same (i.e., $k' = k$) and vary over the range 3--5. As in earlier plots, we compute relative reductions for the methods. Since $k.R^\epsilon_{rand}(kb)$ is only an estimate, if this quantity exceeds the original value of the metric, we set it to the original value and its relative reduction to 100\%. For simplicity, we refer to this quantity as $UB$.

RR* denotes relative reduction in the sum of ``long'' paths. Figs.~\ref{fig:PCS_VC100k} and~\ref{fig:PCS_VC200k} present $UB$ and the relative reduction (RR*) for $PCS$. RR* by PCS is within $50\%$ of UB. (As $kb$ is unknown, PCS may in fact occasionally produce a higher RR* than UB, Fig. \ref{fig:PCS_VC100k} for $k'=3$). Increasing the length threshold ($k'$) reduces the difference between UB and PCS.  

Finally, we explore the reduction by PCS in the sum of long versus short paths. Fig.~\ref{fig:PCS_VC100} compares the reduction: RR for the actual metric (sum of all paths) and RR* for the restricted metric (sum of paths of length $\geq k'$). There is a higher chance for long paths to contain upgraded vertices. As expected, the figure shows higher reduction when the metric includes only longer paths.

\vspace{2mm}
\textbf{Experiments: Effect of parameters} 
\newline
\begin{figure}[t]
\vspace{-3mm}
    \centering
    \subfloat[RR]{\includegraphics[width=0.25\textwidth]{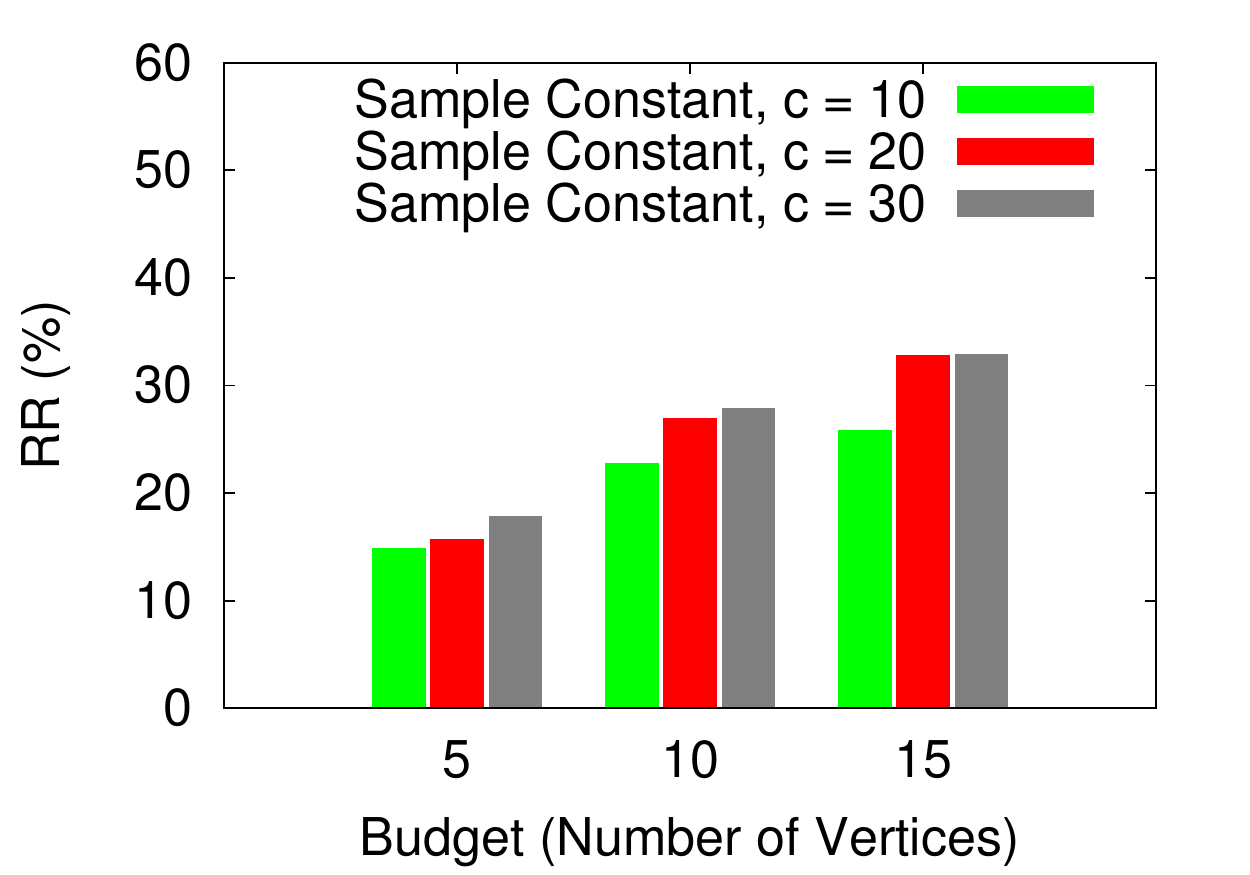}\hspace{-0.2cm}\label{fig:twitter_celeb_qual_sample}}
\subfloat[Time]{\includegraphics[width=0.25\textwidth]{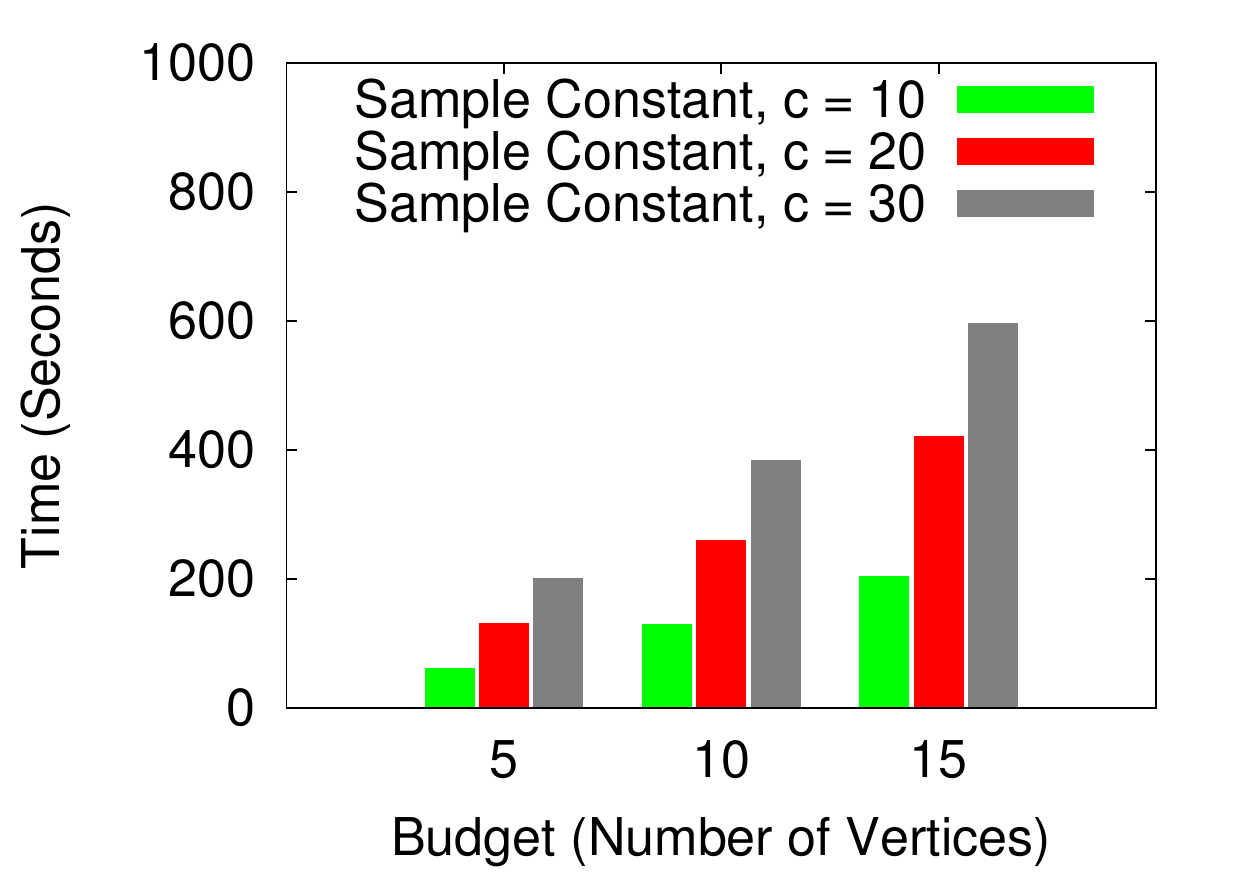}\hspace{-0.2cm}\label{fig:twitter_celeb_time_sample}}
\vspace{-2mm}
    \caption{ (a-b) Quality of GS as a function of the number of samples on Twitter-Celeb.\label{fig:PCS_others_1}}
    \vspace{-3mm}
\end{figure}

The most important parameter for our sampling schemes is the number of samples. In Figs.~\ref{fig:twitter_celeb_qual_sample} and~\ref{fig:twitter_celeb_time_sample} we present the variation in quality and performance of GS with the increase of number of samples on the Twitter-Celeb data. The running time of GS grows linearly with the sample constant $c$, while the quality increases and then saturates confirming our main premise that not all SP need to be observed to make a good-quality selection in our design problem.

We also study the effect of different inputs on our algorithms on synthetic and real networks and their scalability with the number of samples and for increasing budget.
The delays in the general model are randomly distributed in $[500,1000]$. Fig.~\ref{fig:PCS_GS_Scale} shows the scalability of our methods for increasing network size (Barabasi graphs, growth parameter $3$). As expected, PCS scales better than GS (on networks with $0.1$ million vertices PCS is $30$ times faster), and both scale significantly better than the non-sampling alternative GR. This experimentally confirms the theoretical running times of Alg. 1,2 and 3. 

Figs. \ref{fig:PCS_time_sample} and \ref{fig:GS_time_sample} present the running time for increasing $c$ and budget=10 for PCS and GS respectively. Note that the number of samples used by our techniques is controlled by $c$: \#samples=$c*log(n)$. As expected, GS and PCS scale linearly with $c$. The same behaviour persists for increasing budget in Figs. \ref{fig:PCS_time_budget} and \ref{fig:GS_time_budget} ($c$ =$15$). These results also confirm the expected theoretical running time behavior.


\end{document}